\documentclass[10pt]{article}

\usepackage[utf8]{inputenc}
\usepackage{amsmath,amssymb,amsthm}
\usepackage{graphicx}
\usepackage{tikz}
\usepackage{subfigure}
\usepackage{hyperref}
\usepackage{array}
\usepackage[textwidth=16cm,textheight=23cm]{geometry}

\newtheorem{theorem}{Theorem}
\newtheorem{lemma}[theorem]{Lemma}
\newtheorem{claim}{Claim}

\newtheorem{proposition}[theorem]{Proposition}
\newtheorem*{proposition2}{Proposition}
\newtheorem{observation}[theorem]{Observation}
\newtheorem{question}[theorem]{Question}
\newtheorem{construction}[theorem]{Construction}
\newtheorem*{definition}{Definition}
\newtheorem{corollary}[theorem]{Corollary}
\newtheorem{conjecture}[theorem]{Conjecture}

\newcommand{\M}{\gamma^{\text{\tiny{ID}}}}
\newcommand{\E}{\mathbb{E}}
\newcommand{\Bin}{\mbox{Bin}}
\renewcommand{\int}{\mathrm{g}}
\newcommand{\Ebig}{E_{\mbox{{\sc big}}}}
\newcommand{\Egood}{E_{\mbox{{\sc ic}}}}

\title{ Bounds for identifying codes\\ in terms of degree parameters\footnote{\noindent This research was supported by the ANR Project IDEA 
- Identifying coDes in Evolving grAphs, ANR-08-EMER-007, 2009-2011. Both authors thank the hospitality and support of the graduate research training
network "Methods for Discrete Structures" in Berlin, Germany, where part of the research took place. The second author wants to thank the FPU grant from the Ministerio de Educaci\'on de España.}} 
\author{Florent Foucaud\footnote{\noindent
LaBRI - Universit\'e de Bordeaux,
351 cours de la Lib\'eration,
33405 Talence cedex, France.
}
\and 
Guillem Perarnau\footnote{\noindent
MA4 - Universitat Polit\`ecnica de Catalunya,
C/ Jordi Girona 1-3,
08034 Barcelona, Spain.
}}

\begin{document}
\maketitle

\begin{abstract}
  An identifying code is a subset of vertices of a graph such that
  each vertex is uniquely determined by its neighbourhood within the
  identifying code. If $\M(G)$ denotes the minimum size of an
  identifying code of a graph $G$, it was conjectured by F.~Foucaud,
  R.~Klasing, A.~Kosowski and A.~Raspaud that there exists a constant $c$ such that if a connected graph $G$
  with $n$ vertices and maximum degree~$d$ admits an identifying
  code, then $\M(G)\leq n-\tfrac{n}{d}+c$. We use probabilistic
  tools to show that for any~$d\geq 3$, $\M(G)\leq
  n-\tfrac{n}{\Theta(d)}$ holds for a large class of graphs
  containing, among others, all regular graphs and all graphs of
  bounded clique number. This settles the conjecture (up to constants) for
  these classes of graphs. In the general case, we prove $\M(G)\leq
  n-\tfrac{n}{\Theta(d^{3})}$. In a second part, we prove that in any
  graph $G$ of minimum degree~$\delta$ and girth at least~5,
  $\M(G)\leq(1+o_\delta(1))\tfrac{3\log\delta}{2\delta}n$. Using the
  former result, we give sharp estimates for the size of the minimum
  identifying code of random $d$-regular graphs, which is about
  $\tfrac{\log d}{d}n$.
\end{abstract}

\section{Introduction}

Given a graph $G$, an identifying code $\mathcal{C}$ is a dominating set 
such that for any two vertices, their neighbourhoods within $\mathcal{C}$ are
nonempty and distinct. This property can be used to distinguish all
vertices of the graph from each other. Identifying codes have
found applications to various fields since the introduction of this
concept in~\cite{KCL98}. These applications include the location of
threats in facilities using sensors~\cite{DRSTU03}, error-detection
schemes~\cite{KCL98} and routing~\cite{CLST07} in networks, as well as
the structural analysis of RNA proteins~\cite{HKSZ06} (under the denomination of differentiating-dominating sets).

In this paper, we address the question of lower and upper bounds on
the size of an identifying code, thus extending earlier works on such
questions (see e.g.~\cite{M06,CHL07,GM07,FGKNPV10,FKKR10}). We focus
on degree-related graph parameters such as the minimum and maximum
degree, and also study the case of regular graphs. An important part
of the paper is devoted to giving the best possible upper bound for
the size of an identifying code depending on the order and the maximum
degree of the graph, a question raised in~\cite{F09}. We also give
improved bounds for graphs of large girth and study identifying codes
in random regular graphs. The main tools used herein are
probabilistic.

We begin by giving our notations and defining the concepts used
throughout the paper.

As \emph{graphs} and unless specifically mentioned, we understand
simple, undirected and finite graphs. The vertex set of a graph $G$ is
denoted by $V(G)$ and its edge set $E(G)$. We also denote its order by
$n=|V(G)|$. The maximum degree of $G$ will be denoted by $d=d(G)$, its
minimum degree, by $\delta=\delta(G)$, and its average degree, by
$\overline{d}=\overline{d}(G)$. We denote by $u\sim v$, the adjacency
between two vertices $u$ and $v$, and by $u\not\sim v$, their
non-adjacency. The set of neighbours of some vertex $v$ is called its
\emph{open neighbourhood} and denoted by $N(v)$, whereas the set of
its neighbours and itself is called its \emph{closed neighbourhood}
and denoted by $N[v]$. If two distinct vertices $u,v$ are such that
$N[u]=N[v]$, they are called \emph{twins}. If $N(u)=N(v)$ but
$u\not\sim v$, $u$ and $v$ are called \emph{false twins}. The
symmetric difference between two sets $A$ and $B$ is denoted by
$A\Delta B$. We use $\log(x)$ to denote the natural logarithm of
$x$. We also make use of the standard asymptotic notations
$o,O,\Theta,\Omega$ and $\omega$.  Usually the asymptotics are taken
either on variables $d,\delta$ or $n$.  We use the notations $o_x(1)$
and $O_x(1)$ to stress the fact that the asymptotic is taken on
variable $x$. If we write $o(1)$ or $O(1)$, then by convention the
asymptotic is taken only on $n$, the number of vertices.

Given a graph $G$ and a subset $\mathcal{C}$ of vertices of $G$, $\mathcal{C}$ is called a
\emph{dominating set} if each vertex of $V(G)\setminus \mathcal{C}$ has at least
one neighbour in $\mathcal{C}$. Set $\mathcal{C}$ is called a \emph{separating set} of $G$ if for each pair $u,v$ of vertices of $G$,
$N[u]\cap \mathcal{C}\neq N[v]\cap \mathcal{C}$ (equivalently, $(N[u]\Delta N[v])\cap
\mathcal{C}\neq\emptyset$). We have the following definition:

\begin{definition}
  Given a graph $G$, a subset of vertices of $V(G)$ which is both a
  dominating set and a separating set is called an \emph{identifying
    code} of $G$.
\end{definition}

First of all it must be stressed that not every graph can have an
identifying code. Observe that a graph containing twin vertices does
not admit any separating set or identifying code. In fact a graph
admits an identifying code if and only if it is \emph{twin-free},
i.e. it has no pair of twins (one can see that if $G$ is twin-free,
$V(G)$ is an identifying code of $G$). Note that if for three distinct
vertices $u,v,w$ of a twin-free graph $G$, $N[u]\Delta N[v]=\{w\}$,
then $w$ belongs to any identifying code of $G$. In this case we say
that $w$ is \emph{$uv$-forced}, or simply \emph{forced}. Observe that
any isolated vertex must belong to any identifying code for the reason
that it must be dominated. For example, an edgeless graph needs all
the vertices in any identifying code. Hence, the bounds of this paper
only hold for graphs with few isolated vertices. In order to shorten
the statements of our results, we assume that all considered graphs
do not have any isolated vertices.

For a given graph, the problem of finding a minimum identifying code
is known to be NP-hard~\cite{CHL03}, even in graphs having small maximum degree and
high girth (to be precise, in planar graphs of maximum degree~4
having arbitrarily large girth~\cite{A10} and planar graphs of maximum
degree~3 and girth at least~9~\cite{ACHL10}).

The minimum size of an identifying code in a graph $G$, denoted
$\M(G)$, is the \emph{identifying code number} of $G$. It is known that for a
twin-free graph $G$ on $n$ vertices having at least one edge, we have:
$$\lceil\log_2(n+1)\rceil\le\M(G)\leq n-1$$ The lower bound is proved
in~\cite{KCL98} and the upper bound, in~\cite{GM07}. Both bounds are
tight and all graphs reaching these two bounds have been classified
(see~\cite{M06} for the lower bound and~\cite{FGKNPV10} for the upper
bound).

When considering graphs of given maximum degree $d$, it was shown
in~\cite{KCL98} that the lower bound can be improved to
$\M(G)\geq\tfrac{2n}{d+2}$. This bound is tight and a
classification of all graphs reaching it has been proposed
in~\cite{F09}. For any $d$, these graphs include some
regular graphs and graphs of arbitrarily large girth.

It was conjectured in~\cite{FKKR10} that the following upper bound holds.

\begin{conjecture}[\cite{FKKR10}]\label{conj}
  There exists a constant $c$ such that for any nontrivial connected twin-free graph $G$ of maximum
  degree~$d$, $\M(G)\leq n-\tfrac{n}{d}+c$.
\end{conjecture}

Graphs of maximum degree~$d$ such that $\M(G) = n -\frac n d$ are
known (e.g. the complete bipartite graph $K_{d,d}$ and richer classes
of graphs described in Section~\ref{apdx:constructions}). Therefore if
Conjecture~\ref{conj} holds, there would exist a constant $c$ such
that for any twin-free graph $G$ on $n$ vertices and of maximum degree~$d$ we
would have $\tfrac{2}{d+2}n\leq\M(G)\leq n-\tfrac{n}{d}+c$, with both
bounds being tight.

Note that Conjecture~\ref{conj} holds for graphs of maximum degree~$2$
(see~\cite{GMS06}). It was shown in~\cite{FGKNPV10} that $\M(G)\leq
n-\tfrac{n}{\Theta(d^5)}$, and $\M(G)\leq n- \tfrac{n}{\Theta(d^3)}$
when $G$ has no forced vertices (in particular, this is true when $G$
is regular). It is also known that the conjecture holds in an
asymptotic way if $G$ is triangle-free: then, $\M(G)\leq
n-\tfrac{n}{d(1+o_d(1))}$~\cite{FKKR10}.

Identifying codes have been previously studied in two models of random
graphs, that is the classic random graph model~\cite{FMMRS07}
and the model of random geometric graphs~\cite{MS09}. To our
knowledge random regular graphs have not been studied in the context
of identifying codes.
\vspace{0.3cm}

In this paper, we further study Conjecture~\ref{conj} and prove that it is
tight (up to constants) for large enough values of~$d$ and for a large
class of graphs, including regular graphs and graphs of bounded clique
number (Corollaries~\ref{cor:constantf(G)}
and~\ref{cor:boundedcliques}). In the general case, we prove that
$\M(G)\leq n- \tfrac{n}{\Theta(d^3)}$
(Corollary~\ref{cor:general}). These results improve the known bounds
given in~\cite{FGKNPV10} and support Conjecture~\ref{conj}. Moreover,
we show that the much improved upper bound of $\M(G)\leq
(1+o_{\delta}(1))\frac{3\log{\delta}}{2\delta}n$ holds for graphs
having girth at least~5 and minimum degree~$\delta$
(Theorem~\ref{the:g5}). This bound is used to give an asymptotically
tight bound of about $\frac{\log{d}}{d}n$ for the identifying code
number of almost all random $d$-regular graphs
(Corollary~\ref{cor:lb}).

We summarize our results for the special case of regular graphs in
Table~\ref{tab:results} and compare them to the bound for the
dominating set problem (the table contains references for both the
bound and its tightness). All bounds are asymptotically tight. We note
that identifying codes behave far from dominating sets in general, as
shown by the first lines of the table: there are regular graphs having
much larger identifying code number than domination number. However,
for larger girth and for almost all regular graphs, the bounds for the
two problems coincide asymptotically, as shown by the last lines of
the table.

\begin{table}[ht!]
\centering
    \begin{tabular}{ccc}
 		&  Identifying codes & Dominating sets\\ 
\hline
\hline
 & & \\[-0.2cm]
     in general &  $n-\tfrac{n}{103d}$ & $\sim\tfrac{\log d}{d}n$ \\
             &  Thm.~\ref{thm:bigthm}, Constr.~\ref{constr2}  & \cite{AS00}, \cite{TY07}\\
 & & \\[-0.3cm]
    \hline
 & & \\[-0.2cm]
     girth~4 &  $n-\tfrac{n}{d(1+o_d(1))}$ & $\sim\tfrac{\log d}{d}n$ \\
             &  \cite{FKKR10}, Constr.~\ref{constr3} & \cite{AS00}, \cite{TY07} \\
 & & \\[-0.3cm]
    \hline
 & & \\[-0.2cm]
     girth~5 &  $(1+o_d(1))\tfrac{3\log d}{2d}n$ & $\sim\tfrac{\log d}{d}n$ \\
             &  Thm.~\ref{the:g5}, Thm.~\ref{prop:domination} & \cite{AS00}, \cite{TY07} \\
 & & \\[-0.3cm]
    \hline
 & & \\[-0.2cm]
     almost all graphs &  $\tfrac{\log{d}+\log\log d+O_d(1)}{d}n$ &
$\sim\tfrac{\log d}{d}n$\\
      &  Thm.~\ref{prop:RRG}, Thm.~\ref{prop:domination} & \cite{AS00}, \cite{TY07}\\
 & & \\[-0.3cm]
    \end{tabular}
\caption{Summary of the upper bounds for $d$-regular graphs}
\label{tab:results}
\end{table}

In order to prove our results, we use probabilistic techniques. For
some results, we use the weighted version of Lov\'asz' Local Lemma to
show the existence of an identifying code, together with the Chernoff
bound to show that this code is small enough.
To bound the number of forced vertices in a graph we study an auxiliary directed graph that
captures the underlying structure of these vertices. This new technique we introduce can be useful
to study the number of forced vertices in a more general context, which is an important problem in
the community of identifying codes.
We also make use of other probabilistic techniques such as the Alteration
Method~\cite{AS00} in order to give better bounds in more restricted
cases. Finally, we work with the Configuration Model~\cite{B01} in
order to compute the identifying code number of almost all random regular graphs.

\vspace{0.3cm}

The organization of this paper is as follows. In
Section~\ref{sec:prel} we state some preliminary results which will be
used throughout the paper. In Section~\ref{sec:mainresult}, we improve
the known upper bounds on the identifying code number of graphs of maximum
degree~$d$. This gives new large families
of graphs for which Conjecture~\ref{conj} holds (up to constants). In
Section~\ref{sec:girth5}, we give an upper bound for graphs having
minimum degree~$\delta$ and girth at least~$5$. In Section~\ref{sec:RRG}, we
give sharp bounds for the identifying code number of almost all
$d$-regular graphs. A further section is dedicated to various
constructions of families of graphs which show the tightness of some
of our results (Section~\ref{apdx:constructions}).

\section{Preliminary results}\label{sec:prel}

We first recall a well-known probabilistic tool: the Lov\'asz Local
Lemma. We use its weighted version, a particularization of the general
version where each event has an assigned weight. The proof can be
found in~\cite{MR01}.

\begin{lemma}[Weighted Local Lemma~\cite{MR01}]
\label{lem:WLL}
Let $\mathcal{E}=\left\{E_1,\ldots,E_M\right\}$ be a set of (typically ``bad'')
 events such that each $E_i$ is mutually independent of
 $\mathcal{E}\setminus(\mathcal{D}_i\cup\left\{E_i\right\})$ where
 $\mathcal{D}_i\subseteq \mathcal E$. Suppose that there exist some integer
 weights $t_1,\ldots,t_M\ge 1$ and a real $p\leq\tfrac{1}{4}$ such that for
 each $1\le i\le M$:
\begin{itemize}
\item $\Pr(E_i)\le p^{t_i}$, and
\item $\sum_{E_j\in \mathcal{D}_i}(2p)^{t_j}\le\frac{t_i}{2}$
\end{itemize}

Then $\Pr(\bigcap_{i=1}^M\overline{E_i})\ge\prod_{i=1}^M(1-(2p)^{t_i})>0$.
\end{lemma} 

Note that in Lemma~\ref{lem:WLL}, since $p\leq\tfrac{1}{4}$ and $(1-x)\geq e^{-(2\log 2)x}$ in
$x\in [0,1/2]$, we have:
\begin{equation}\label{eq:expbound}
\textstyle \Pr(\bigcap_{i=1}^M\overline{E_i})\geq\exp \left\{- (2\log 2)\sum_{i=1}^M
(2p)^{t_i} \right\}.
\end{equation}

We also use the following version of the well-known Chernoff bound,
which is a reformulation of Theorem A.1.13 in~\cite{AS00}.

\begin{theorem}[Chernoff bound~\cite{AS00}]\label{Chernoff}
Let $X$ be a random variable of $n$ independent trials of probability
 $p$, and let $a>0$ be a real number. Then $\Pr(X-np\le-a)\le
 e^{-\frac{a^2}{2np}}$.
\end{theorem} 

The following observation gives an equivalent condition for a set to
be an identifying code, and follows from the fact that for two
vertices $u,v$ at distance at least~3 from each other,
$N[u]\Delta N[v]=N[u]\cup N[v]$.

\begin{observation}\label{obs:dist2}
For a graph $G$ and a set $\mathcal{C}\subseteq V(G)$, if $\mathcal{C}$ is dominating and
$N[u]\cap \mathcal{C}\neq N[v]\cap \mathcal{C}$ for each pair of vertices $u,v$ at
distance at most two from each other, then $N[u]\cap \mathcal{C}\neq N[v]\cap \mathcal{C}$
for each pair of vertices of the graph.
\end{observation}

The next observation is immediate, but it is worth mentioning here.

\begin{observation}
Let $G$ be a twin-free graph and $\mathcal{C}$, an identifying code of $G$. Any
set $\mathcal{C}'$ such that  $\mathcal{C}\subseteq \mathcal{C}'$ is also an identifying code of $G$.
\end{observation}

The next proposition shows an upper bound on the number of false
twins in a graph.

\begin{proposition}\label{prop:falsetwins}
  Let $G$ be a graph on $n$ vertices having maximum degree~$d$ and no
  isolated vertices, then $G$ has at most $\tfrac{n(d-1)}{2}$ pairs of
  false twins.
\end{proposition}
\begin{proof}
  Let us build a graph $H$ on $V(G)$, where two vertices $u,v$ are
  adjacent in $H$ if they are false twins in $G$. Note that since a
  vertex can have at most $d-1$ false twins, $H$ has maximum degree
  $d-1$. Therefore it has at most $\tfrac{n(d-1)}{2}$ edges and the
  claim follows.
\end{proof}

Note that the bound of Proposition~\ref{prop:falsetwins} is tight since
in a complete bipartite graph $K_{d,d}$, $n=2d$ and there are exactly
$2\binom{d}{2}=\tfrac{n(d-1)}{2}$ pairs of false twins.

\section{Upper bounds on the identifying code number}\label{sec:mainresult}

\subsection{Main theorem}

In this section, we improve the known upper bounds of~\cite{FGKNPV10} on
the identifying code number by using the Weighted Local Lemma, stated
in Lemma~\ref{lem:WLL}.

In the following, given a graph $G$ on $n$ vertices, we will denote by
$f(G)$ the proportion of non-forced vertices of $G$, i.e. the ratio
$\tfrac{x}{n}$, where $x$ is the number of non-forced vertices of $G$.

\begin{theorem}
\label{thm:bigthm}
  Let $G$ be a twin-free graph
  on $n$ vertices having maximum degree~$d\geq 3$. Then $\M(G) \leq
  n-\frac{nf(G)^2}{103d}$.
\end{theorem}

\begin{proof}
  Let $F$ be the set of forced vertices of $G$, and $V'=V(G)\setminus
  F$. Note that $|V'|=nf(G)$. By the definition of a forced vertex,
  any identifying code must contain all vertices of $F$.
	
In this proof, we first build a set $S$ in a random manner by choosing
vertices from $V'$. Then we exhibit some ``bad'' configurations --- if
none of those occurs, the set $\mathcal{C}=F\cup (V'\setminus S)$ is an identifying
code of $G$. Using the Weighted Local Lemma, we compute a lower bound on the
(non-zero) probability that none of these bad events occurs. Finally, we
use the Chernoff bound to show that with non-zero probability, the
size of $S$ is also large enough for our purposes. This shows
that such a ``good'' large set $S$ exists, and it can be used to build an identifying code
that has a sufficiently small size.

Let $p=p(d)$ be a probability which will be determined
later. We build the set $S\subseteq V'$ such that each vertex of $V'$
independently belongs to $S$ with probability $p$. Therefore the
random variable $|S|$ follows a binomial distribution \Bin($nf(G)$,$p$)
and has expected value $\E(|S|)=p n f(G)$.

Let us now define the set $\mathcal{E}$ of ``bad'' events. These are
of four types. An illustration of these events is given in
Figure~\ref{fig:configs}.

\begin{itemize}
\item \textbf{Type A$^j$ ($2\leq j\leq d+1$):} for each vertex $u\in V'$, let $A^j_u$ be the event that $|N[u]=j|$ and $N[u]\subseteq
      S$.
\item \textbf{Type B$^j$ ($2\leq j\leq 2d-2$):} for each pair $\left\{u,v\right\}$ of adjacent
vertices, let $B_{u,v}^j$ be the event that $|(N[u]\Delta N[v])|=j$ and
$(N[u]\Delta N[v])\subseteq S$. 
\item \textbf{Type C$^j$ ($3\leq j\leq 2d$):} for each pair $\left\{u,v\right\}$ of
	  vertices in $V'$ at distance two from each other, let $C_{u,v}^j$ be the event that
$|(N[u]\Delta N[v])|=j$ and $(N[u]\Delta N[v])\subseteq S$.      
\item \textbf{Type D:} for each pair $\left\{u,v\right\}$ of false twins in $V'$, let
$D_{u,v}$ be the event that
      $(N[u]\Delta N[v])=\{u,v\}\subseteq S$.
\end{itemize}

For the sake of simplicity, we refer to the events of type $A^j$,
$B^j$ and $C^j$ as events of type $A$, $B$ and $C$ respectively
whenever the size of the symmetric difference is not relevant.

Events of type $B_{u,v}^1$ are not defined since then $|N[u]\Delta
N[v]|=1$ and $F$ belongs to the code, so they never happen. Observe
that the events $C^j_{u,v}$ and $D_{u,v}$ are just defined over the
pairs of vertices in $V'$ because if either $u$ or $v$ belongs to $F$,
the event does not happen.

If no event of type $A$ occurs, $V(G)\setminus S$ is a dominating set of
$G$. If no event of type $B$ occurs, all pairs of adjacent
vertices are separated by $V(G)\setminus S$. If no event of type $C$ or $D$
occurs, all pairs of vertices at distance~2 from each other are
separated. Thus by Observation~\ref{obs:dist2}, if no event of type $A$, $B$, $C$ or $D$ occurs,
then $V(G)\setminus S$ is also a separating set of $G$, and therefore it is an identifying code of
$G$.

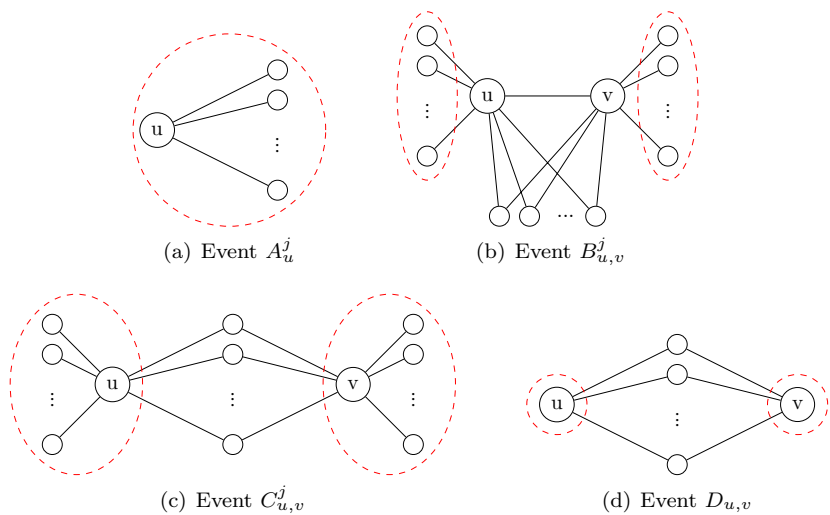
\begin{figure}[!htpb]
\centering

\subfigure[Event $A^j_u$]{\label{fig:config-a}
\scalebox{0.8}{\begin{tikzpicture}
\path (0,0) node[draw,shape=circle] (u) {u};
\path (2,1) node[draw,shape=circle] (a1) {};
\path (2,-1) node[draw,shape=circle] (a2) {};
\path (2,0.5) node[draw,shape=circle] (a3) {};
\draw (u) -- (a1)
      (u) -- (a2)
      (u) -- (a3);
\draw[red,dashed] (1.2,0) ellipse (1.6cm and 1.6cm);
\draw (2,-0.25) node[rotate=90] {...};
\end{tikzpicture}}

}\qquad
\subfigure[Event $B^j_{u,v}$]{\label{fig:config-b}
\scalebox{0.8}{{\begin{tikzpicture}
\path (-1,1) node[draw,shape=circle] (a1) {};
\path (-1,0.5) node[draw,shape=circle] (a2) {};
\path (-1,-1) node[draw,shape=circle] (a3) {};
\path (3,1) node[draw,shape=circle] (b1) {};
\path (3,0.5) node[draw,shape=circle] (b2) {};
\path (3,-1) node[draw,shape=circle] (b3) {};
\path (0,0) node[draw,shape=circle] (u) {u};
\path (2,0) node[draw,shape=circle] (v) {v};
\path (0.2,-2) node[draw,shape=circle] (n1) {};
\path (1.8,-2) node[draw,shape=circle] (n2) {};
\path (0.7,-2) node[draw,shape=circle] (n3) {};
\draw (u) -- (n1)
      (v) -- (n1)
      (u) -- (n2)
      (v) -- (n2)
      (u) -- (n3)
      (v) -- (n3)
      (a1) -- (u)
      (a2) -- (u)
      (a3) -- (u)
      (v) -- (b1)
      (v) -- (b2)
      (v) -- (b3)
      (u) -- (v);
\draw[red,dashed] (-1,0) ellipse (0.5cm and 1.4cm)
           (3,0) ellipse (0.5cm and 1.4cm);
\draw (1.3,-2) node {...}
      (-1,-0.25) node[rotate=90] {...}
      (3,-0.25) node[rotate=90] {...};
\end{tikzpicture}}
}
}\\
\subfigure[Event $C^j_{u,v}$]{\label{fig:config-c}
\scalebox{0.8}{\begin{tikzpicture}
\path (-1,1) node[draw,shape=circle] (a1) {};
\path (-1,0.5) node[draw,shape=circle] (a2) {};
\path (-1,-1) node[draw,shape=circle] (a3) {};
\path (5,1) node[draw,shape=circle] (b1) {};
\path (5,0.5) node[draw,shape=circle] (b2) {};
\path (5,-1) node[draw,shape=circle] (b3) {};
\path (0,0) node[draw,shape=circle] (u) {u};
\path (4,0) node[draw,shape=circle] (v) {v};
\path (2,1) node[draw,shape=circle] (n1) {};
\path (2,-1) node[draw,shape=circle] (n2) {};
\path (2,0.5) node[draw,shape=circle] (n3) {};
\draw (u) -- (n1)
      (v) -- (n1)
      (u) -- (n2)
      (v) -- (n2)
      (u) -- (n3)
      (v) -- (n3)
      (a1) -- (u)
      (a2) -- (u)
      (a3) -- (u)
      (v) -- (b1)
      (v) -- (b2)
      (v) -- (b3);
\draw[red,dashed] (-0.6,0) ellipse (1.1cm and 1.5cm)
           (4.6,0) ellipse (1.1cm and 1.5cm);
\draw (2,-0.25) node[rotate=90] {...}
      (-1,-0.25) node[rotate=90] {...}
      (5,-0.25) node[rotate=90] {...};
\end{tikzpicture}}
}\qquad
\subfigure[Event $D_{u,v}$]{\label{fig:config-d}
\scalebox{0.8}{\begin{tikzpicture}
\path (0,0) node[draw,shape=circle] (u) {u};
\path (4,0) node[draw,shape=circle] (v) {v};
\path (2,1) node[draw,shape=circle] (n1) {};
\path (2,-1) node[draw,shape=circle] (n2) {};
\path (2,0.5) node[draw,shape=circle] (n3) {};
\draw (u) -- (n1)
      (v) -- (n1)
      (u) -- (n2)
      (v) -- (n2)
      (u) -- (n3)
      (v) -- (n3);
\draw[red,dashed] (u) ellipse (0.5cm and 0.5cm)
           (v) ellipse (0.5cm and 0.5cm);
\draw (2,-0.25) node[rotate=90] {...};
\end{tikzpicture}}
}
\caption{The ``bad'' events. The vertices in dashed circles belong
 to set $S$.}
\label{fig:configs}
\end{figure}

Let $V(E_i)$ denote the set of vertices that must belong to set $S$ so that $E_i$ holds (see
Figure~\ref{fig:configs}, where the sets $V(E_i)$ are the ones inside the dashed circles). We will say that a vertex $v\in V(G)$ participates to $E_i$, if
$v\in V(E_i)$. We define the weight $t_i$ of each event $E_i\in\mathcal{E}$ as $|V(E_i)|$. For $j\geq
2$ and for $T\in\{A^j,B^j,C^j,D\}$, let $t_T$ be the weight of an event of type
$T$ (for an event $E_i\in\mathcal E$ of type $T$, $t_i=t_T$). We have
the following:
 $$t_{A^j}=j\qquad t_{B^j}=j\qquad t_{C^j}=j \qquad t_D=2$$

Some vertex $x$ can participate to at most $d+1$ events of type $A$ since if it participates to some
event $A^j_u$, then $u\in N[x]$. Vertex $x$ can participate to at most $d(d-1)$ events of type $B$:
supposing $x\in V(B^j_{u,v})$ and $u$ is adjacent to $x$, there are at most $d$ ways to choose $u$,
and at most $d-1$ ways to choose $v$ among $N(u)\setminus \{x\}$. Observe that if $x=u$ or $x=v$,
then $x\notin V(B_{u,v})$ (see Figure~\ref{fig:config-b}). Similarly $x$ can participate to at most
$d^2(d-1)$ events of type $C$: for some event $C^j_{u,v}$, there are at most $d(d-1)$ possibilities
if $x=u$ or $x=v$ and at most $d(d-1)^2$ if $u$ or $v$ is a neighbour of $x$. Finally, $x$ can
participate to at most $d-1$ events $D_{u,v}$ since $x$ can have at most $d-1$ false twins. For each
type $T$ of events ($T\in\{A^j,B^j,C^j,D\}$) and any vertex $v\in V(G)$, let us define $\int(v,T)$
to be the number of events $E_i$ of type $T$ such that $v\in V(E_i)$. Hence:

\begin{equation}\label{eq:ints}
\begin{split}
&\sum_{j=2}^{d+1}\int(v,A^j)\leq d+1 \qquad \sum_{j=2}^{2d-2} \int(v,B^j)\leq d(d-1) \\
&\sum_{j=3}^{2d} \int(v,C^j)\leq d^2(d-1) \qquad \int(v,D)\leq d-1
\end{split}
\end{equation}

Let us call $\Egood$ the event that no event of $\mathcal E$ occurs. Using the Weighted Local Lemma,
we want to show that $\Pr(\Egood)>0$. Given two events $E_i$ and $E_j$ of $\mathcal{E}$, we note
$i\sim j$ if $V(E_i)\cap V(E_j)=\emptyset$. Observe that for any event $E_i$ and any set
$T\subseteq \{j:\, i\not\sim j\}$, we have $\Pr(E_i \mid \cap_{j\in T} \overline{E_j})=\Pr(E_i)$,
since the vertices are included in $S$ with independent probabilities. This means that $E_i$ is
mutually independent from the set of all events $E_j$ for which $V(E_i)\cap V(E_j)=\emptyset$.

In order to apply the Weighted Local Lemma (Lemma~\ref{lem:WLL}), the
following conditions must hold for each event $E_i\in\mathcal E$:

\begin{equation*}
\sum_{i\sim j} (2p)^{t_j}\leq \frac{t_i}{2}
\end{equation*}

The latter conditions are implied by the following ones (for each event $E_i\in\mathcal E$):

\begin{equation*}
\begin{split}
\sum_{j=2}^{d+1}\sum_{v\in V(E_i)}&\int(v,A^j)(2p)^{t_{A^j}}+\sum_{j=2}^{2d-2}\sum_{v\in
V(E_i)}\int(v,B^j)(2p)^{t_{B^j}}+\\
&\sum_{j=3}^{2d} \sum_{v\in
V(E_i)}\int(v,C^j)(2p)^{t_{C^j}}+\sum_{v\in V(E_i)}\int(v,D)(2p)^{t_{D}}\leq\frac{t_i}{2}
\end{split}
\end{equation*}

Which are implied by:

\begin{equation*}
\begin{split}
t_i\cdot\max\limits_{v\in V(E_i)}&\left\{\sum_{j=2}^{d+1}\int(v,A^j)(2p)^{t_{A^j}}\right\}+
t_i\cdot\max\limits_{v\in V(E_i)}\left\{\sum_{j=2}^{2d-2}\int(v,B^j)(2p)^{t_{B^j}}\right\}+\\
&t_i\cdot\max\limits_{v\in V(E_i)}\left\{\sum_{j=3}^{2d}\int(v,C^j)(2p)^{t_{C^j}}\right\}
+t_i\cdot\max\limits_{v\in V(E_i)}\left\{\int(v,D)(2p)^{t_{D}}\right\}\leq\frac{t_i}{2}
\end{split}
\end{equation*}

Using the bounds of Inequalities~\eqref{eq:ints} and noting that for $p\leq 1/4$ and
any $j$, $(2p)^{t_{A^j}}\leq (2p)^2$, $(2p)^{t_{B^j}}\leq (2p)^2$ and
$(2p)^{t_{C^j}}\leq (2p)^3$, for any event $E_i$ this equation is implied by:

\begin{equation}
\label{eq:WLL}
(d+1)(2p)^2+d(d-1)(2p)^2+d^2(d-1)(2p)^3+(d-1)(2p)^2= 4d^2p^2+8d^3p^3+4dp^2-8d^2p^3\leq \frac{1}{2}
\end{equation}

Hence, we fix $p=\tfrac{1}{kd}$ where $k$ is a constant to be determined later.
Equation~\eqref{eq:WLL} holds for $k\geq 3.68$ for all $d\geq 3$. In fact, in the following steps of the proof,
we will assume that $k\geq 30$, and so Equation~\eqref{eq:WLL} will be satisfied for any $d\geq 3$.
Since $p\leq\tfrac{1}{4}$ and $\Pr(E_i)\leq p^{t_i}$ by the definition of $t_i$ and the choice of
$S$, the Weighted Local Lemma can be applied.

Let $M_T$ be the number of events of type $T$, where
$T\in\left\{A^j,B^j,C^j,D\right\}$. By Lemma~\ref{lem:WLL} we have:

$$
\Pr(\Egood)\geq\prod_{j=2}^{d+1}\prod_{i=1}^{M_{A^j}}(1-(2 p)^{t_{A^j}})\prod_{j=2}^{2d-2}\prod_{i=1}^{M_{B^j}}(1-(2
 p)^{t_{B^{j}}})\prod_{j=3}^{2d}\prod_{i=1}^{M_{C^j}}(1-(2 p)^{t_{C^j}})\prod_{i=1}^{M_D}(1-(2 p)^{t_D})
$$

Note that $\sum_{j=2}^{d+1}M_{A^j}= n f(G)$ since by definition there exists exactly one event
$A^j_u$ for each vertex of $u\in V'$. Moreover, $\sum_{j=2}^{2d-2}M_{B^j}\leq\frac{nd}{2}$ since
there is exactly one event type $B^j_{u,v}$ for each edge $uv\in E(G)$ and at most $\frac{nd}{2}$
edges in $G$. We also have that $\sum_{j=3}^{2d}M_{C^j}$ is at most the number of
pairs of vertices in $V'$ at distance~$2$ from each other. This is
also at most the number of paths of length~2 with both endpoints in
$V'$, which is upper-bounded by $\frac{n f(G) d
  (d-1)}{2}$. Finally, $M_D$ is the number of pairs of false
twins in $V'$, which is at most $n f(G)\tfrac{d-1}{2}$ by
Proposition~\ref{prop:falsetwins}. Hence, we have:

$$
\Pr(\Egood)  \geq  (1-(2 p)^2)^{nf(G)} (1-(2 p)^2)^{\tfrac{nd}{2}}(1-(2p)^3)^{\tfrac{n
f(G)d(d-1)}{2}}(1-(2p)^2)^{\tfrac{nf(G) (d-1)}{2}}
$$

Using Lemma~\ref{lem:WLL} (more precisely, we use
Equation~\eqref{eq:expbound}) and the fact that $p=\tfrac{1}{k
  d}$, we obtain:

\begin{align*}
\Pr(\Egood)  \geq&  \exp\left\{-(2\log 2)(2p)^2
\left(f(G)+\frac{d}{2}+\frac{f(G)d(d-1)2p}{2}+\frac{f(G)(d-1)}{2}
\right)n\right\}\\
\geq&  \exp\left\{-\frac{4\log
2}{k^2d}\left(\frac{2f(G)}{d}+1+\frac{2f(G)}{k}+f(G)\right)n\right\}
\end{align*}

Since $f(G)\leq 1$ and it is assumed that $k\geq 30$,
one can check that for any $d\geq 3$:\footnote{Note that this bound could be strengthened by assuming $d$ to be large
enough. Indeed, here the term $\tfrac{2f(G)}{d}$ can be as high as
$\tfrac{2}{3}$ when $d=3$ and $f(G)=1$, but can be chosen to be as low
as desired by assuming $d$ to be larger. However we aim at giving a
bound for any $d\geq 3$, hence we use the weaker bound presented here.}
$$
\Pr(\Egood)\ge \exp\left\{-\frac{164\log 2}{15k^2d}n\right\}
$$

The Weighted Local Lemma shows that $S$ has the desired properties
with probability $\Pr(\Egood)>0$, implying that such a set
exists. Note that we have no guarantee on the size of $S$. In fact, if
$S=\emptyset$ then $V(G)\setminus S=V(G)$ is always an identifying
code. Therefore we need to estimate the probability that $|S|$ is
far below its expected size. In order to do this, we use the
Chernoff bound of Theorem~\ref{Chernoff} by putting
$a=\tfrac{nf(G)}{cd}$ where $c$ is a constant to be determined. Let $\Ebig$ be the event that
$|S|-np>-\tfrac{nf(G)}{cd}$. We obtain:

\begin{eqnarray*}
\Pr(\overline{\Ebig}) & \le & \exp\left\{-\frac{\left(\tfrac{nf(G)}{cd}\right)^2}{2pnf(G)}\right\}\\
& = & \exp\left\{-\frac{kf(G)}{2c^2d}n\right\}
\end{eqnarray*}

Now we have:
\begin{eqnarray*}
\Pr(\Egood \text{~and~} \Ebig) & = & 1-\Pr(\overline{\Egood}\text{~or~} \overline{\Ebig})\\
&\ge & 1-\Pr(\overline{\Egood})-\Pr(\overline{\Ebig})\\
& = & 1-(1-\Pr(\Egood))-\Pr(\overline{\Ebig})\\
& = & \Pr(\Egood)-\Pr(\overline{\Ebig})\\
& \ge & \exp\left\{-\frac{164\log 2}{15k^2d}n \right\}-\exp\left\{ -\frac{kf(G)}{2c^2d}n\right\}
\end{eqnarray*}

Thus, $\Pr(\Egood \text{~and~} \Ebig)>0$ if
$c<\tfrac{k^{3/2}f(G)^{1/2}}{\sqrt{\tfrac{328\log 2}{15}}}$. We (arbitrarily) set
$c=\tfrac{k^{3/2}f(G)^{1/2}}{\sqrt{22\log 2}}$ in order to fulfill this condition.

Now we have to check that $\Ebig$ implies that $S$ is still large enough.

\begin{eqnarray}
|S| & \geq & \mathbb{E}(|S|)-\frac{nf(G)}{cd} \nonumber \\
& = & \frac{nf(G)}{kd}-\frac{nf(G)}{cd} \nonumber \\
& = & \left(\frac{1}{k} - \frac{\sqrt{22\log 2}}{k^{3/2}f(G)^{1/2}} \right)
\frac{nf(G)}{d}\label{finaleqn}
\end{eqnarray}

Since $|S|$ must be positive, from Equation~\eqref{finaleqn} we need
$k^{3/2}f(G)^{1/2} > \sqrt{22\log 2}\,k$, which leads to $k =
\tfrac{a_0}{f(G)}$ for $a_0>22\log 2$. Using all our previous assumptions, by derivating the expression of
$|S|$, one can check that $|S|$ is maximized when
$a_0=\tfrac{99\log 2}{2}$. Hence we set $k=\tfrac{99\log 2}{2f(G)}$.

Remark that under this condition and since $f(G)\leq 1$, we have $k\geq 34$
and our assumption following Equation~\eqref{eq:WLL} that $k\geq 30$,
is fulfilled.

Now, with $a_0=\tfrac{99\log 2}{2}$, we can see that:
$$
|S| \geq \left(\frac{1}{k}-\frac{1}{c}\right)\frac{nf(G)}{d}=\frac{a_0^{1/2}-\sqrt{22\log 2}}{a_0^{3/2}}\frac{f(G)^2}{d}n=\frac{2}{297\log
2}\frac{f(G)^2}{d}n\geq\frac{f(G)^2}{103d}n
$$

Hence finally the identifying code $\mathcal{C}=V\setminus S$ has size
$$
|\mathcal{C}| \leq n-\frac{nf(G)^2}{103d}
$$
\end{proof}

Note that for regular graphs, $f(G)=1$ because a forced vertex implies
the existence of two vertices with distinct degrees. We obtain the following result:

\begin{corollary}[Graphs with constant proportion of non-forced vertices]\label{cor:constantf(G)}
   Let $G$ be a twin-free graph
  on $n$ vertices having maximum degree~$d\geq 3$ and
  $f(G)=\tfrac{1}{\alpha}$ for some constant $\alpha\geq 1$. Then $\M(G)\le
  n-\frac{n}{103\alpha^2 d}$. In particular if $G$ is $d$-regular,
  $\M(G)\leq n-\frac{n}{103d}$.
\end{corollary}

The next proposition will be proved in the next subsection.

\begin{proposition}\label{prop:f(G)}
Let $G$ be a graph on $n$ vertices and of maximum degree $d$. Then
$f(G)\geq\tfrac{1}{d+1}$.
\end{proposition}

We obtain the following general result:

\begin{corollary}[General case]\label{cor:general}
   Let $G$ be a twin-free graph
  on $n$ vertices having maximum degree~$d\geq 3$. Then $\M(G)\le
  n-\frac{n}{103d(d+1)^2}= n-\frac{n}{\Theta(d^3)}$.
\end{corollary}

The next proposition will be proved in the next subsection as well.

\begin{proposition}\label{prop:forced_planar}
Let $G$ be a graph having no $k$-clique. Then there exists a constant
$\gamma(k)$ depending only on $k$, such that
$f(G)\geq\tfrac{1}{\gamma(k)}$.
\end{proposition}

This leads to the following extension of
Corollary~\ref{cor:constantf(G)}, where $c(k)\leq 103\gamma(k)^2$:

\begin{corollary}[Graphs with bounded clique number]\label{cor:boundedcliques}
  There exists an integer $d_0$ such that for each twin-free graph $G$
  on $n$ vertices having maximum degree~$d\geq d_0$ and clique number
  smaller than~$k$, $\M(G)\le n-\frac{n}{c(k)d}$ for some constant
  $c(k)$ depending only on $k$. In particular this applies to
  triangle-free graphs, planar graphs, or more generally, graphs of
  bounded genus.
\end{corollary}

We remark here that the previous corollaries support
Conjecture~\ref{conj}. They also lead us to think that the difficulty of
the problem lies in forced vertices.

\subsection{Bounding the number of non-forced vertices: proofs}

In this section, we prove the lower bounds for function $f(G)$ of the
statement of Theorem~\ref{thm:bigthm}. 

The following lemma was first proved in~\cite{Ber01}, and a proof can be
found in~\cite{FGKNPV10} (as~\cite{Ber01} is not accessible).

\begin{lemma}[\cite{Ber01}]\label{lemma:takeout_vertex}
  If $G$ is a finite twin-free graph without isolated vertices, then
  for every vertex $u$ of $G$, there is a vertex $v \in N[u]$ such
  that $V(G)\setminus\{v\}$ is an identifying code of $G$.
\end{lemma}

We recall the statement of Proposition~\ref{prop:f(G)}:

\begin{proposition2}
Let $G$ be a graph on $n$ vertices and of maximum degree $d$. Then
$f(G)\geq\tfrac{1}{d+1}$.
\end{proposition2}
\begin{proof}
Observe that a vertex $v$ of $G$ is not forced only if $V(G)\setminus\{v\}$
is an identifying code of $G$. Hence, by
Lemma~\ref{lemma:takeout_vertex}, the set $S$ of non-forced vertices is a
dominating set of $G$, and thus $|S|\geq\tfrac{n}{d+1}$.
\end{proof}

Note that Proposition~\ref{prop:f(G)} is tight. Indeed, consider the
graph $A_k$ on $2k$ vertices defined in~\cite{FGKNPV10} as follows:
$V(A_k)=\{x_1,\ldots, x_{2k}\}$ and $E(A_k)=\{x_ix_j, |i-j|\leq
k-1\}$. $A_k$ can be seen as the $(k-1)$-th power of the path
$P_{2k}$. In the graph $A_k$ with an additional universal vertex $x$
(i.e. $x$ is adjacent to all vertices of $A_k$), one can
check that all vertices but $x$ are forced. This graph has $n=2k+1$
vertices, maximum degree~$2k$ and exactly
$1=\tfrac{n}{d+1}$ non-forced vertex. Taking all forced vertices
gives a minimum identifying code of this graph.

However, note that since for a fixed even value of $d$, we know only one
such graph, it is not enough to give a counterexample to
Conjecture~\ref{conj}. Indeed in this case the size of the code is 
$n-1= n-\frac{n}{d+1}=n-\tfrac{n}{d}+\frac{1}{n-1}= n-\tfrac{n}{d}+1$. So we ask the following
question:

\begin{question}
Does there exist a value of $d$ such that for an infinite number of values of $n$ there exists a
graph on $n$ vertices of maximum degree $d$ having exactly $\tfrac{n}{d+1}$ non-forced vertices?
\end{question}

Answering this question in positive would provide counterexamples to
Conjecture~\ref{conj}.  Note that for the similar question where we
replace $d+1$ by $d$, the answer is positive by
Construction~\ref{constr1} of Section~\ref{apdx:constructions}. For
any $d$, this construction provides arbitrarily large graphs having exactly
$\tfrac{n}{d}$ non-forced vertices.

Observe that graph $A_k$ contains two cliques of $k$ vertices. In
fact, we can improve the bound of Proposition~\ref{prop:f(G)} for
graphs having no large cliques. Let us first introduce an
auxiliary structure that will be needed in order to prove this result.

Let $G$ be a twin-free graph. We define a partial order $\preceq$ over
the set of vertices of $G$ such that $u\preceq v$ if $N[u]\subseteq
N[v]$. We construct an oriented graph $\mathcal{H}(G)$ on $V(G)$ as a
subgraph of the Hasse diagram of poset $(V(G),\preceq)$. The arc set of
$\mathcal{H}(G)$ is the set of all arcs $\overrightarrow{uv}$ where
there exists some vertex $x$ such that $N[v]=N[u]\cup\{x\}$. 
Then $x$ is $uv$-forced, and we note $x=f(\overrightarrow{uv})$.  For
a vertex $v$ of $V(G)$, we define the set $F(v)$ as the union of $v$
itself and the set of all predecessors and successors of $v$ in
$\mathcal{H}(G)$. Observe that $\mathcal{H}(G)$ has no directed cycle since it represents a partial
order, and thus predecessors and successors are well-defined.

\begin{lemma}\label{lemma:bounded-F}
Let $G$ be a graph having no $k$-clique. Then for each vertex $u$,
$|F(u)|\leq \beta(k)$, where $\beta(k)$ is a function depending only
on $k$.
\end{lemma}
\begin{proof}
First of all, we prove that the maximum in-degree of $\mathcal{H}(G)$
is at most $2k-3$, and its out-degree is at most~$k-2$. 

Let $u$ be a vertex of $G$. Suppose $u$ has $2k-2$ in-neighbours in $\mathcal{H}(G)$. Since for each
in-neighbour $v$ of $u$, $|N[u]\Delta N[v]|=1$ in $G$, each of them is non-adjacent in $G$ to at
most one of the other in-neighbours (in the worst case the in-neighbours of $u$ induce in $G$ a
clique of $2k-2$ vertices minus the edges of a perfect matching). Hence they induce a clique of size
at least~$k-1$ in $G$. Together with vertex $u$, they form a $k$-clique in $G$, a contradiction.

Now suppose $u$ has $k-1$ out-neighbours in
$\mathcal{H}(G)$. Since for each out-neighbour $v$ of $u$ in
$\mathcal{H}(G)$, $N[u]\subseteq N[v]$ in $G$, $u$ and its out-neighbours form a $k$-clique in $G$,
a contradiction.

Now, consider the subgraph of $\mathcal{H}(G)$ induced by $F(u)$. We
claim that the longest directed chain in this subgraph has at most
$k-1$ vertices. Indeed, all the vertices of such a chain are pairwise
adjacent in $G$. Since $G$ is assumed not to have any $k$-cliques, there are at most $k-1$
vertices in a directed chain.

Finally, we obtain that $F(u)$ has size at most
$\beta(k)=\sum_{i=0}^{k-2}(2k-3)^i$ and the claim of the lemma follows.
\end{proof}

We now need to prove a few additional claims regarding the structure
of $\mathcal{H}(G)$. In the following claims, we suppose that $G$ is a
twin-free graph.

\begin{claim}\label{clm:f(G)-2}
  Let $s$ be a forced vertex in $G$ with $s=f(\overrightarrow{uv})$
  for some vertices $u$ and $v$. If $t$ is an in-neighbour of $s$ in
  $\mathcal{H}(G)$, then $v=f(\overrightarrow{ts})$. Moreover if $v$ is forced with
  $v=f(\overrightarrow{xy})$, then necessarily $y=s$.
\end{claim}
\begin{proof}
  For the first implication, suppose $s$ has an in-neighbour $t$ in
  $\mathcal{H}(G)$. An illustration is provided in
  Figure~\ref{fig:claimA}. Since $u\not\sim s$, then $u\not\sim
  t$. Moreover $v\not\sim t$ since $s=f(\overrightarrow{uv})$. Since
  $s\sim v$ the claim follows. For the other implication, suppose
  there exist two vertices $x,y$ such that
  $v=f(\overrightarrow{xy})$. Hence $y\sim v$ but $x\not\sim
  v$. Therefore $u\not\sim x$ (otherwise $v$ would be adjacent to $x$
  too) and hence $u\not\sim y$. Now the only vertex adjacent to $v$
  but not to $u$ is $s$, so $y=s$.
\end{proof}

\begin{figure}[!htpb]
\centering
\scalebox{0.8}{\begin{tikzpicture}[join=bevel,inner sep=0.6mm,line width=0.8pt,scale=1]
\path (0,0) node[draw,shape=circle,fill] (u) {};
  \path (u)+(-0.3,-0.3) node {$u$};
\path (2,0) node[draw,shape=circle,fill] (t) {};
  \path (t)+(0.3,-0.3) node {$t$};
\path (2,2) node[draw,shape=circle,fill] (s) {};
  \path (s)+(0.8,0.3) node {$s=f(\overrightarrow{uv})$};
\path (0,2) node[draw,shape=circle,fill] (v) {};
  \path (v)+(-0.8,0.3) node {$v=f(\overrightarrow{ts})$};
\draw[->, line width=1.2pt] (u) -- (v);
\draw[->, line width=1.2pt] (t) -- (s);
\draw[blue, line width=0.5pt] (v) -- (s);
\draw[-,dashed, blue, line width=0.5pt] (v) -- (t)
                       (u) -- (s);
\end{tikzpicture}}
\caption{The situation of Claim~\ref{clm:f(G)-2}. Arcs belong to
  $\mathcal{H}(G)$. Full thin edges belong to $G$ only, dashed edges
  are non-edges in $G$.}
\label{fig:claimA}
\end{figure}
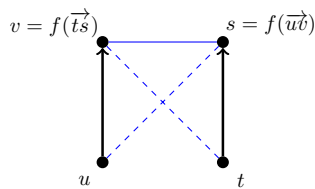

\begin{claim}\label{clm:f(G)-3}
Let $s$ be a forced vertex in $G$ with $s=f(\overrightarrow{uv})$ for
some vertices $u$ and $v$. Then $s$ has at most one in-neighbour in
$\mathcal{H}(G)$.
\end{claim}
\begin{proof}
Suppose $s$ has two distinct in-neighbours $t$ and $t'$ in
$\mathcal{H}(G)$ (see Figure~\ref{fig:claimB} for an illustration). By
Claim~\ref{clm:f(G)-2}, $v$ is both $ts$-forced and $t's$-forced. But
then $N[t]=N[s]\setminus\{v\}=N[t']$. Then $t$ and $t'$ are twins, a
contradiction since $G$ is twin-free.
\end{proof}

\begin{figure}[!htpb]
\centering
\scalebox{0.8}{\begin{tikzpicture}[join=bevel,inner sep=0.6mm,line width=0.8pt,scale=1]
\path (0,0) node[draw,shape=circle,fill] (u) {};
  \path (u)+(-0.3,-0.3) node {$u$};
\path (1.5,0) node[draw,shape=circle,fill] (t) {};
  \path (t)+(0.3,-0.3) node {$t$};
\path (2.5,0) node[draw,shape=circle,fill] (t2) {};
  \path (t2)+(0.3,-0.3) node {$t'$};
\path (2,2) node[draw,shape=circle,fill] (s) {};
  \path (s)+(0.8,0.3) node {$s=f(\overrightarrow{uv})$};
\path (0,2) node[draw,shape=circle,fill] (v) {};
  \path (v)+(-0.3,0.3) node {$v$};
\draw[->, line width=1.2pt] (u) -- (v);
\draw[->, line width=1.2pt] (t) -- (s);
\draw[->, line width=1.2pt] (t2) -- (s);
\draw[blue, line width=0.5pt] (v) -- (s);
\draw[-,dashed, blue, line width=0.5pt] (v) -- (t)
                      (v) -- (t2)
                      (u) -- (s);
\end{tikzpicture}}
\caption{The situation of Claim~\ref{clm:f(G)-3}. Arcs belong to
  $\mathcal{H}(G)$. Full thin edges belong to $G$ only, dashed edges
  are non-edges in $G$.}
\label{fig:claimB}
\end{figure}
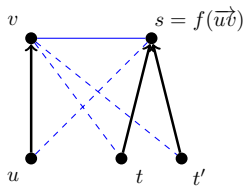

\begin{claim}\label{clm:f(G)-4}
Let $s$ be a forced vertex in $G$ with $s=f(\overrightarrow{uv})$, and
let $t$ be a forced in-neighbour of $s$ in $\mathcal{H}(G)$ with
$t=f(\overrightarrow{xy})$ for some vertices $u,v,x,y$. Then $x=v$.
\end{claim}
\begin{proof}
Since $t\sim y$, then $s\sim y$ too. But since $t=f(\overrightarrow{xy})$,
$x\sim s$ and $x\not\sim t$. Now by Claim~\ref{clm:f(G)-2},
$v=f(\overrightarrow{ts})$, that is, $v$ is the unique vertex such
that $v$ is adjacent to $s$, but not to $t$. Therefore $x=v$.
\end{proof}

We now obtain the following lemma using the previous claims.

\begin{lemma}\label{clm:f(G)-5}
Let $s$ be a non-isolated sink in $\mathcal{H}(G)$ which is forced in
$G$ with $s=f(\overrightarrow{uv})$ for some vertices $u$ and
$v$. Then either $s$ has a non-forced predecessor $t$ in
$\mathcal{H}(G)$ such that $F(s)\subseteq F(t)$, or there exists a
non-forced vertex $w(s)$ such that $F(s)\subseteq
N_G[w(s)]$. Moreover, if there are $\ell$ additional sinks
$\{s_1,\ldots,s_\ell\}$ which are all non-isolated in $\mathcal{H}(G)$ and such that $w(s)=w(s_1)=\ldots=w(s_\ell)$, then
there exists a set of $\ell+1$ distinct vertices inducing a clique together with $w(s)$.
\end{lemma}
\begin{proof}
First of all, recall that $\mathcal{H}(G)$ has no directed circuits.
Suppose $s$ has a non-forced predecessor in $\mathcal{H}(G)$ and let
$t$ be one such predecessor having the shortest distance to $s$ in $\mathcal{H}(G)$. By
Claim~\ref{clm:f(G)-3}, predecessors of $s$ are either successors or predecessors of $t$, and
there is a directed path from $t$ to $s$ in $\mathcal{H}(G)$. Hence $F(s)\subseteq
F(t)$, which proves the first part of the statement.

Now suppose all predecessors of $s=f(\overrightarrow{uv})$ are forced. By
Claim~\ref{clm:f(G)-3}, $s$ and its predecessors form a directed path
$\{t_0,\ldots,t_m,s\}$ in $\mathcal{H}(G)$ (for an illustration, see
Figure~\ref{fig:lemma16-a}). Note that by Claim~\ref{clm:f(G)-2},
we have $v=f(\overrightarrow{t_ms})$. By our
assumption we know that $t_m$ is forced, say
$t_m=f(\overrightarrow{xv_m})$ for some vertices $x$ and $v_m$. But
now by Claim~\ref{clm:f(G)-4}, $x=v$ and
$t_m=f(\overrightarrow{vv_m})$. Now, repeating these arguments for
each other predecessor of $s$ shows that there is a directed path
$\{u,v,v_m,\ldots,v_0\}$ with $t_m=f(\overrightarrow{vv_m})$ and for
all $i$, $0\le i\le m-1$, $t_i=f(\overrightarrow{v_{i+1}v_i})$. In
particular, $t_0=f(\overrightarrow{v_1v_0})$.
Observe also that for all $i\geq 1$, $v_{i}=f(\overrightarrow{t_{i-1}t_{i}})$.
By applying Claim~\ref{clm:f(G)-4} on vertices $v_1,v_0$ and $t_0$, if $v_0$ is
forced then $t_0$ has an in-neighbour in $\mathcal{H}(G)$, a
contradiction --- hence $v_0$ is non-forced. Moreover note that since
$v_0\sim t_0$, then $v_0$ is adjacent to all successors of $t_0$ in
$\mathcal{H}(G)$, that is, to all elements of $F(s)$. Therefore,
putting $w(s)=v_0$, we obtain the second part of the statement.

For the last part, suppose there exists a set of $\ell$ additional forced sinks $\{s_1,\ldots,s_\ell\}$ which are non-isolated in $\mathcal{H}(G)$ and such that all
their predecessors in $\mathcal{H}(G)$ are forced with $w(s_i)=v_0$ for $1\leq i\leq \ell$ (for an
illustration, see Figure~\ref{fig:lemma16-b}). For each such sink $s_i$, by the previous paragraph,
the vertices of $F(s_i)$ induce a directed path $\{t^i_0,\ldots,t^i_{m_i},s_i\}$ in
$\mathcal{H}(G)$. Moreover we know that there is a vertex $x_i$ such that $t^i_0$ is
$x_iv_0$-forced. We claim that the set of vertices $X=\{x_1,\ldots,x_\ell\}$ together with $v_0$ and
$v_1$, form a clique in $G$ of $\ell+2$ vertices. 

We first claim that for all $i,j$ in $\{1,\ldots,\ell\}$, $x_i\neq
t_0^j$. If $i=j$, this is clear by our assumptions. Otherwise, suppose
by contradiction, that $x_i=t_0^j$ for some $i\neq j$ in
$\{1,\ldots,\ell\}$. Then we claim that $x_j=t_0^i$. Indeed, by the
previous part of the proof, we know that
$f(\overrightarrow{x_jv_0})=t_0^j=x_i$ --- hence $x_j\not\sim
x_i$. But since $\overrightarrow{x_iv_0}$ is an arc in
$\mathcal{H}(G)$, we must have
$f(\overrightarrow{x_iv_0})=x_j$. Again, we know that
$f(\overrightarrow{x_iv_0})=t_0^i$, hence $x_j=t_0^i$. Let $t_1^i$
denote the successor of $t_0^i$ in the directed path from $t_0^i$ to
$s_i$ in $\mathcal{H}(G)$. We know from the previous part of the proof
that $f(\overrightarrow{t_0^it_1^i})=x_i=t_0^j$. However since
$t_0^i=x_j$ we also know that $f(\overrightarrow{t_0^iv_0})=x_i$. This
implies that $N_G[v_0]=N_G[t_1^i]$, a contradiction since these two
vertices are distinct and $G$ is twin-free.

Now, observe that the vertices of $X$ must all be pairwise adjacent.
All vertices of $X$ are adjacent to $v_0$, and for each $x_i$,
$N[v_0]=N[x_i]\cup\{t^i_0\}$, hence $x_i$ is adjacent to all
neighbours of $v_0$ except $t^i_0$. But by the previous paragraph, we
know that $t^i_0\neq x_j$ for all $j\in\{1,\ldots,\ell\}$, hence $x_i$
is adjacent to all $x_j\neq x_i$, $j\in\{1,\ldots,\ell\}$. For the
same reason, each $x_i$ is adjacent to $v_1$. Hence, the vertices of
$X$ form a clique together with $v_0$ and $v_1$.

Finally, let us show that all the vertices of $X$ are distinct: by
contradiction, suppose that $x_i=x_j$ for some $i\neq j$, $1\leq
i,j\leq \ell$. Since $t^i_0$ is $x_iv_0$-forced and $t^j_0$ is
$x_jv_0$-forced, we have $t^i_0=t^j_0$. Since $s_i$ and $s_j$ are
distinct, this means that $s_i$ and $s_j$ have one predecessor in
common. Hence their common predecessor which is nearest to $s_i$ and
$s_j$, say $t$, has two out-neighbours. Let $t_i$ (respectively $t_j$)
be the out-neighbour of $t$ which is a predecessor of $s_i$
(respectively $s_j$) --- see Figure~\ref{fig:lemma16-c} for an
illustration. We know that there are two vertices $y_i,y_j$ such that
$y_i=f(\overrightarrow{tt_i})$ and
$y_j=f(\overrightarrow{tt_j})$. First note that $y_i$ and $y_j$ are
distinct: otherwise, we would have
$N[t_i]=N[t]\cup\{y_i\}=N[t]\cup\{y_j\}=N[t_j]$ and then $t_i, t_j$
would be twins in $G$. Observe that since $t\not\sim
y_i$ and $y_i\neq f(\overrightarrow{tt_j})$, we have $t_j\not\sim y_i$. We
know that $t$ is forced, in fact by the first part of this proof, we also know
that $t=f(\overrightarrow{y_iz_i})$ for some vertex $z_i$. Hence
$z_i\sim t$, and since $N[t]\subseteq N[t_j]$, $z_i\sim t_j$. But
since $t_j\neq f(\overrightarrow{y_iz_i})$, $t_j\sim y_i$, a
contradiction. Hence $x_i$ and $x_j$ are distinct, which completes the
proof.
\end{proof}

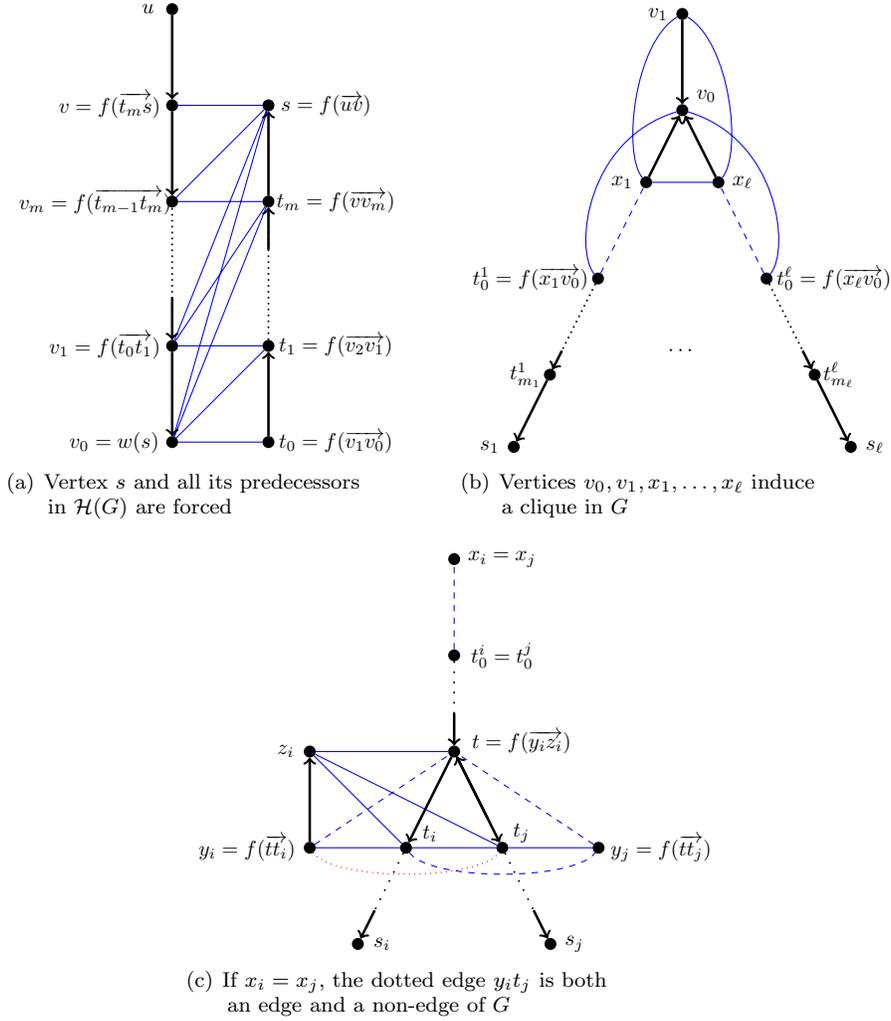
\begin{figure}[!htpb]
\centering
\subfigure[Vertex $s$ and all its predecessors\newline\hspace*{1.5em} in $\mathcal{H}(G)$ are forced]{\label{fig:lemma16-a}
\scalebox{0.8}{\begin{tikzpicture}[join=bevel,inner sep=0.6mm,line width=0.8pt,scale=0.8]
\path (0,0) node[draw,shape=circle,fill] (v0) {};
  \path (v0)+(-1.2,0) node {$v_0=w(s)$};
\path (0,2) node[draw,shape=circle,fill] (v1) {};
  \path (v1)+(-1.4,0) node {$v_1=f(\overrightarrow{t_0t_1})$};
\path (0,5) node[draw,shape=circle,fill] (vk) {};
  \path (vk)+(-1.6,0) node {$v_m=f(\overrightarrow{t_{m-1}t_m})$};
\path (0,7) node[draw,shape=circle,fill] (v) {};
  \path (v)+(-1.3,0) node {$v=f(\overrightarrow{t_ms})$};
\path (0,9) node[draw,shape=circle,fill] (u) {};
  \path (u)+(-0.5,0) node {$u$};

\path (2,0) node[draw,shape=circle,fill] (t0) {};
  \path (t0)+(1.4,0) node {$t_0=f(\overrightarrow{v_1v_0})$};
\path (2,2) node[draw,shape=circle,fill] (t1) {};
  \path (t1)+(1.4,0) node {$t_1=f(\overrightarrow{v_2v_1})$};
\path (2,5) node[draw,shape=circle,fill] (tk) {};
  \path (tk)+(1.4,0) node {$t_m=f(\overrightarrow{vv_m})$};
\path (2,7) node[draw,shape=circle,fill] (s) {};
  \path (s)+(1.2,0) node {$s=f(\overrightarrow{uv})$};

\draw[dotted] (v1) -- (vk)
              (t1) -- (tk);

\draw[blue, line width=0.5pt] (v0) -- (t0)
            (v0) -- (t1)
            (v1) -- (t1)
            (vk) -- (tk)
            (vk) -- (s)
            (v) -- (s)
            (v1) -- (s)
            (v0) -- (s)
            (v1) -- (tk)
            (v0) -- (tk);

\draw[->, line width=1.2pt] (u) -- (v);
\draw[->, line width=1.2pt] (v) -- (vk);
\draw[->, line width=1.2pt] (vk)+(0,-2) -- (v1);
\draw[->, line width=1.2pt] (v1) -- (v0);
\draw[->, line width=1.2pt] (t0) -- (t1);
\draw[->, line width=1.2pt] (t1)+(0,2) -- (tk);
\draw[->, line width=1.2pt] (tk) -- (s);
\end{tikzpicture}}
}\qquad
\subfigure[Vertices $v_0,v_1,x_1,\ldots,x_\ell$ induce\newline\hspace*{1.5em} a clique in $G$]{\label{fig:lemma16-b}
\scalebox{0.8}{\begin{tikzpicture}[join=bevel,inner sep=0.6mm,line width=0.8pt,scale=0.8]
\path (0,0) node[draw,shape=circle,fill] (v0) {};
  \path (v0)+(0.5,0.3) node {$v_0$};
\path (0,2) node[draw,shape=circle,fill] (v1) {};
  \path (v1)+(-0.5,0) node {$v_1$};
\path (-0.75,-1.5) node[draw,shape=circle,fill] (x1) {};
  \path (x1)+(-0.5,0) node {$x_1$};
\path (-1.75,-3.5) node[draw,shape=circle,fill] (t01) {};
  \path (t01)+(-1.4,0) node {$t_0^1=f(\overrightarrow{x_1v_0})$};
\path (-2.75,-5.5) node[draw,shape=circle,fill] (tm11) {};
  \path (tm11)+(-0.5,0) node {$t_{m_1}^1$};
\path (-3.5,-7) node[draw,shape=circle,fill] (s1) {};
  \path (s1)+(-0.5,0) node {$s_1$};

\path (0.75,-1.5) node[draw,shape=circle,fill] (x2) {};
  \path (x2)+(0.5,0) node {$x_\ell$};
\path (1.75,-3.5) node[draw,shape=circle,fill] (t02) {};
  \path (t02)+(1.4,0) node {$t_0^\ell=f(\overrightarrow{x_\ell v_0})$};
\path (2.75,-5.5) node[draw,shape=circle,fill] (tm22) {};
  \path (tm22)+(0.5,0) node {$t_{m_\ell}^\ell$};
\path (3.5,-7) node[draw,shape=circle,fill] (s2) {};
  \path (s2)+(0.5,0) node {$s_\ell$};

\path (0,-5) node {$\cdots$};

\draw[dotted] (t01) -- (tm11)
              (t02) -- (tm22);

\draw[blue, line width=0.5pt] (t01) .. controls +(-0.5,0.6) and +(-2,-0.6) .. (v0)
            (t02) .. controls +(0.5,0.6) and +(2,-0.6) .. (v0)
            (x1) .. controls +(-0.5,0.6) and +(-1,-0.6) .. (v1)
            (x2) .. controls +(0.5,0.6) and +(1,-0.6) .. (v1)
            (x1) -- (x2);

\draw[->, line width=1.2pt] (v1) -- (v0);
\draw[->, line width=1.2pt] (x1) -- (v0);
\draw[->, line width=1.2pt] (t01)+(-0.75,-1.5) -- (tm11);
\draw[->, line width=1.2pt] (tm11) -- (s1);
\draw[->, line width=1.2pt] (x2) -- (v0);
\draw[->, line width=1.2pt] (t02)+(0.75,-1.5) -- (tm22);
\draw[->, line width=1.2pt] (tm22) -- (s2);

\draw[-,dashed, blue, line width=0.5pt] (t01) -- (x1)
                      (t02) -- (x2);
\end{tikzpicture}}
}\\
 \subfigure[If $x_i=x_j$, the dotted edge $y_it_j$ is both\newline\hspace*{1.5em} an edge and a non-edge
of $G$]{\label{fig:lemma16-c}
 \scalebox{0.8}{\begin{tikzpicture}[join=bevel,inner sep=0.6mm,line width=0.8pt,scale=0.8]
\path (0,0) node[draw,shape=circle,fill] (t) {};
  \path (t)+(1.4,0.2) node {$t=f(\overrightarrow{y_iz_i})$};
\path (0,2) node[draw,shape=circle,fill] (t0) {};
  \path (t0)+(1.0,0) node {$t_0^i=t_0^j$};
\path (0,4) node[draw,shape=circle,fill] (xi) {};
  \path (xi)+(1.0,0) node {$x_i=x_j$};

\path (-1,-2) node[draw,shape=circle,fill] (ti) {};
  \path (ti)+(0.5,0.3) node {$t_i$};
\path (-3,-2) node[draw,shape=circle,fill] (yi) {};
  \path (yi)+(-1.3,0) node {$y_i=f(\overrightarrow{tt_i})$};
\path (-3,0) node[draw,shape=circle,fill] (zi) {};
  \path (zi)+(-0.5,0) node {$z_i$};
\path (-2,-4) node[draw,shape=circle,fill] (si) {};
  \path (si)+(0.5,0) node {$s_i$};
\path (1,-2) node[draw,shape=circle,fill] (tj) {};
  \path (tj)+(0.4,0.3) node {$t_j$};
\path (3,-2) node[draw,shape=circle,fill] (yj) {};
  \path (yj)+(1.3,0) node {$y_j=f(\overrightarrow{tt_j})$};
\path (2,-4) node[draw,shape=circle,fill] (sj) {};
  \path (sj)+(0.5,0) node {$s_j$};

\draw[loosely dotted] (ti) -- (si)
              (tj) -- (sj)
              (t0) -- (t);

\draw[blue, line width=0.5pt] (t) -- (zi)
            (zi) -- (tj)
            (zi) -- (ti)
            (yi) -- (ti)
            (ti) -- (tj)
            (yj) -- (tj);

\draw[->, line width=1.2pt] (t)+(0,0.8) -- (t);
\draw[->, line width=1.2pt] (t) -- (ti);
\draw[->, line width=1.2pt] (t) -- (tj);
\draw[->, line width=1.2pt] (yi) -- (zi);
\draw[->, line width=1.2pt] (x2) -- (v0);
\draw[->, line width=1.2pt] (si)+(0.35,0.7) -- (si);
\draw[->, line width=1.2pt] (sj)+(-0.35,0.7) -- (sj);

\draw[-,dashed, blue, line width=0.5pt] (t0) -- (xi)
                      (yi) -- (t)
                      (yj) -- (t)
                      (ti) .. controls +(0.5,-0.7) and +(-0.5,-0.7) .. (yj);

\draw[dotted,red, line width=0.5pt] (yi) .. controls +(0.5,-0.7) and +(-0.5,-0.7) .. (tj);
\end{tikzpicture}}
}

\caption{Three situations in the proof of Lemma~\ref{clm:f(G)-5}. Arcs belong to
  $\mathcal{H}(G)$. Full thin edges belong to $G$ only, dashed edges
  are non-edges in $G$.}
\label{fig:lemma16}
\end{figure}

Finally, let us recall and prove Proposition~\ref{prop:forced_planar}.

\begin{proposition2}
Let $G$ be a graph having no $k$-clique. Then there exists a constant
$\gamma(k)$ depending only on $k$, such that
$f(G)\geq\tfrac{1}{\gamma(k)}$.
\end{proposition2}
\begin{proof}
To prove the result, we use $\mathcal{H}(G)$ to construct a set
$X=\{x_1,\ldots,x_\ell\}$ of non-forced vertices such that
$\bigcup_{i=1}^\ell A(x_i)=V(G)$, where $A(x_i)$ is a set of at most
$\gamma(k)$ vertices. Then we have $\ell\geq\tfrac{n}{\gamma(k)}$ vertices in $X$ and the claim of
the proposition follows.

We now describe a procedure to build set $X$ while considering each
non-isolated sink of $\mathcal{H}(G)$. We denote by $s$ the currently
considered sink.
\vspace{.2cm}

\textbf{Case 1:} Sink $s$ is non-forced. Then we set $A(s)$ to be
$F(s)$ together with all the vertices which are forced by a pair $u,v$
of vertices of $F(s)$. Note that by Lemma~\ref{lemma:bounded-F},
$|F(s)|\leq \beta(k)$, where $\beta(k)$ only depends on $k$. Hence,
$|A(s)|\leq \beta(k)+\binom{\beta(k)}{2}$.
\vspace{.2cm}

\textbf{Case 2:} Sink $s$ is forced. By Lemma~\ref{clm:f(G)-5}, either
$s$ has a non-forced predecessor $t$ such that $F(s)\subseteq F(t)$,
or there exists a non-forced vertex $w(s)$ such that $F(s)\subseteq
N_G[w]$. 

In the first case, we choose $t$ as our non-forced vertex, and we set
$A(t)$ to be $F(t)$ together with all the vertices which are forced by
a pair $u,v$ of vertices of $F(t)$. Again we have $|A(t)|\leq
\beta(k)+\binom{\beta(k)}{2}$.
  
In the second case, we choose $w=w(s)$ as our non-forced vertex. Now,
let $S=\{s,s_1,\ldots,s_\ell\}$ be the set of forced sinks having no
non-forced predecessor and such that $w(s)=w(s_1)=\ldots
w(s_\ell)$. By Lemma~\ref{clm:f(G)-5} we know that there are $\ell+1$
distinct vertices inducing a clique together with $w$, hence $\ell+2< k$.
We set $A(w)$ to be $F(w)\cup F(s)\cup F(s_1)\cup\ldots\cup F(s_\ell)$ together with all the vertices
which are forced by a pair $u,v$ of vertices of this set. We
have $|A(w)|\leq k\beta(k)+\binom{k\beta(k)}{2}$.
\vspace{.2cm}

We have now covered all the vertices which are not isolated in
$\mathcal{H}(G)$, since for each non-isolated sink $s$ of
$\mathcal{H}(G)$, $F(s)$ is a subset of $A(x)$ for some $x\in X$.
Moreover all isolated vertices of $\mathcal{H}(G)$ which are forced,
have also been put into some set $A(x)$. Hence only non-forced
isolated vertices of $\mathcal{H}(G)$ need to be covered. For each
such vertex $v$, we add $v$ to $X$ and set $A(v)=\{v\}$.

Finally, all vertices belong to some set $A(x)$, $x\in X$, and the
size of each set $A(x)$ is at most $\gamma(k)=k\beta(k)+\binom{k\beta(k)}{2}$,
which completes the proof.
\end{proof}

\section{Upper bounds for graphs with girth at least 5}\label{sec:girth5}

This section is devoted to the study of graphs that have girth at
least $5$. We will use these results in Section~\ref{sec:RRG}, which
deals with random regular graphs.

Despite being different than our previous proofs, the ones of this
section have also a probabilistic flavour. One can check that for graphs of
girth~5, applying the Local Lemma does not lead to a satisfying
result. However, by using the Alteration method, a better bound can be
given.

We start by defining an auxiliary notion that will be used in this
section. A subset $D\subseteq V(G)$ is called a \emph{$2$-dominating
  set} if for each vertex $v$ of $V(G)\setminus D$, $|N(v)\cap D|\ge
2$~\cite{FJ85}. The next lemma shows that we can use a $2$-dominating
set to construct an identifying code.

\begin{lemma}\label{lemma:2DS=IC}
  Let $G$ be a twin-free graph on $n$ vertices having girth at
  least~5. Let $D$ be a $2$-dominating set of $G$. If the subgraph
  induced by $D$, $G[D]$, has no isolated edge, $D$ is an identifying
  code of $G$.
\end{lemma}
\begin{proof}
  First observe that $D$ is dominating since it is
  $2$-dominating. Let us check that $D$ is also separating.

  Note that all the vertices that do not belong to $D$
  are separated because they are dominated at least twice each and
  $g(G)>4$.

  Similarly, a vertex $x\in D$ and a vertex $y\in V(G)\setminus D$ are
  separated since $y$ has two vertices which dominate it, but
  they cannot both dominate $x$ (otherwise there would be a triangle or a 4-cycle in
  $G$).

  Finally, consider two vertices of $D$. If they are not adjacent they
  are separated by themselves. Otherwise, by the assumption that
  $G[D]$ has no isolated edge and that $G$ has no triangles, we know
  that at least one of them has a neighbour in $D$, which separates
  them since it is not a neighbour of the other.
\end{proof}

The following theorem makes use of Lemma~\ref{lemma:2DS=IC}. The idea
of the proof is inspired by a classic proof of a result on dominating
sets which can be found in the first chapter of~\cite{AS00}.

\begin{theorem}
\label{the:g5}
Let $G$ be a graph on $n$ vertices with minimum degree $\delta$ and
girth at least $5$. Then
$\M(G)\leq (1+o_{\delta}(1))\frac{3\log{\delta}}{2\delta}n$. Moreover if $G$ has average degree $\overline{d}=O_\delta(\delta(\log{\delta})^2)$ then $\M(G)\leq
\frac{\log{\delta}+\log\log\delta+O_\delta(1)}{\delta}n$. 
\end{theorem}

\begin{proof}[Proof]
  Let $S\subseteq V(G)$ be a random subset of vertices, where each vertex
  $v\in V(G)$ is added to $S$ uniformly at random with probability $p$
  (where $p$ will be determined later). For every vertex $v\in V(G)$, we
  define the random variable $X_v$ as follows:
$$
X_v = \left\{ \begin{array}{rl} 
0 & \mbox{if } |N[v]\cap S|\geq 2 \\ 
1 & \mbox{otherwise}
\end{array}\right. 
$$

Let $T=\{v \mid X_v=1\}$. This set contains, in particular, the subset of vertices which are not $2$-dominated by $S$. Note that $|T|=\sum X_v$. 
Let us estimate the size of $T$. Observing that $|N[v]\cap S|\sim \Bin(\deg(v)+1,p)$ and
$\deg(v)\geq \delta$, we obtain:

\begin{eqnarray*}
\mathbb{E}(|T|) & = & \sum_{v\in V(G)} \mathbb{E}(X_v)\\ 
& \leq & n \left( (1-p)^{\delta+1}+(\delta+1)p(1-p)^\delta \right)\\
& = & n(1-p)^\delta((1-p)+(\delta+1)p)\\
& \leq & n(1+\delta p)e^{-\delta p}.
\end{eqnarray*}
where we have used the fact that $1-x\leq e^{-x}$.
Now, note that the set $D=S\cup T$ is a $2$-dominating set of $G$. 
We have $|D|\leq |S|+|T|$. Hence

\begin{eqnarray}
\mathbb{E}(|D|) & \leq & \mathbb{E}(|S|)+\mathbb{E}(|T|)\nonumber\\
& \leq & np+n(1+\delta p)e^{-\delta p}\label{eq:expD}
\end{eqnarray}

Let us set $p=\tfrac{\log \delta+\log\log\delta}{\delta}$. Plugging this into Equation~\eqref{eq:expD}, we obtain:
$$\mathbb{E}(|D|)\leq\frac{\log\delta+\log\log\delta}{\delta}n+\frac{1+\log\delta+\log\log\delta}{\delta\log\delta}n=\frac{\log\delta+\log\log\delta+O_\delta(1)}{\delta}n$$
This shows that there exists at least one $2$-dominating set $D$
having this size.
\vspace{.2cm}

\textbf{Case 1:} (general case) Note that we can use
Lemma~\ref{lemma:2DS=IC} by considering all pairs $u,v$ of vertices of
$D$ forming an isolated edge in $G[D]$, and add an arbitrary neighbour of either
one of them to $D$. Observe that such a vertex exists, otherwise $u$
and $v$ would be twins in $G$. Since there are at most
$\frac{|D|}{2}$ such pairs, we obtain a $2$-dominating set of size at
most $|D|+\frac{|D|}{2}=
(1+o_{\delta}(1))\frac{3\log{\delta}}{2\delta}n$ having the desired
property. Now applying Lemma~\ref{lemma:2DS=IC} completes Case~$1$.
\vspace{.2cm}

\textbf{Case 2:} (sparse case) Whenever
$\overline{d}=O_\delta(\delta(\log{\delta})^2)$, we can get a
better bound by estimating the number of isolated edges of $G[D]$. For convenience, we define the
random variables $Y_{uv}$ for each edge $uv$ of $G$, as follows:
$$
Y_{uv} = \left\{ \begin{array}{rl} 
1 & \mbox{if $N[u]\Delta N[v]\subseteq V(G)\setminus S$}  \\ 
0 & \mbox{otherwise}
\end{array}\right. 
$$

An isolated edge in $G[D]$ might have been created in several ways. First, at the initial
construction step of $S$: if both $u,v$ belong to $S$, but none of their other neighbours do which
happens with probability at most $p^2(1-p)^{2\delta-2}$. A second possibility is in the step where
we add the vertices of $T$ to our solution. This could happen if both $u,v$ were not dominated at
all by $S$, which occurs with probability at most $(1-p)^{2\delta}$, or if exactly one of $u,v$ was
part of $S$ and none of their neighbours were, which has probability at most $2p(1-p)^{2\delta-1}$.
Thus, the total probability of having an isolated edge in $G[D]$ is bounded from above as follows.
$$\Pr(Y_{uv}=1)\leq p^2(1-p)^{2\delta-2}+(1-p)^{2\delta}+2p(1-p)^{2\delta-1}=(1-p)^{2\delta-2}$$

Using the previous observation together with the facts that $p=\tfrac{\log\delta+\log\log\delta}{\delta}$ and $1-x\leq e^{-x}$, let us
calculate the expected value of $Y=\sum_{uv\in E(G)} Y_{uv}$.

$$
\mathbb{E}(Y)= \sum_{uv\in E(G)} \mathbb{E}(Y_{uv})\leq \frac{n\overline{d}}{2} (1-p)^{2\delta-2}
\leq \frac{n\overline{d}}{2}e^{-(2\delta-2)p} =
\frac{n\overline{d}e^{-2(\log\delta+\log\log\delta)}}{2}
= \frac{n\overline{d}}{2\delta^2(\log\delta)^2}
$$

We construct $U$ by picking an arbitrary neighbour of either $u$ or $v$ for each edge $uv$ such that $Y_{uv}=1$. We have $|U|\leq Y$.
The final set $\mathcal{C}=S\cup T\cup U$ is an identifying code. Now we have:
$$
\mathbb{E}(|\mathcal{C}|)\leq \mathbb{E}(|S|)+\mathbb{E}(|T|)+\mathbb{E}(|U|)\leq \frac{\log{\delta}+\log\log\delta+O_\delta(1)}{\delta}n + \frac{\overline{d}}{2\delta^2(\log\delta)^2}n
$$

Using that $\overline{d}=O_{\delta}(\delta(\log{\delta})^2)$,
\begin{equation}\label{eq:expected}
\mathbb{E}(|\mathcal{C}|)\leq \frac{\log{\delta}+\log\log\delta+O_\delta(1)}{\delta}n 
\end{equation}

Then there exists some choice of $S$ such that $|\mathcal{C}|$ has the
desired size, and completes the proof.
\end{proof}

In fact, it is shown in the next section (Corollary~\ref{cor:lb}) that
Theorem~\ref{the:g5} is asymptotically tight. 

Moreover, note that Theorem~\ref{the:g5} cannot be extended much in
the sense that if we drop the condition on girth~5, we know
arbitrarily large $d$-regular triangle-free graphs having large
minimum identifying codes. For instance, Construction~\ref{constr3} of
Section~\ref{apdx:constructions} provides a graph $G$ which satisfies
$\M(G)=n-\tfrac{n}{d}$. Similarly, we cannot drop the
minimum degree condition. Indeed it is known that any $(d-1)$-ary
complete tree $T_{d,h}$ of height $h$, which is of maximum degree~$d$,
minimum degree~1 and has infinite girth, also has a large identifying
code number (i.e. $\M(T_{d,h})=n-\tfrac{n}{d-1+o_d(1)}$~\cite{BCHL05}).

\section{Identifying codes of random regular graphs}\label{sec:RRG}

From the study of regular graphs arises the question of the value of
the identifying code number for most regular graphs. We know some
lower and upper bounds for this parameter, but is it
concentrated around some value? A good way to study this question is
to look at random regular graphs.

Consider the Configuration Model, where a $d$-regular
multigraph on $n$ vertices is obtained by selecting some perfect
matching of $K_{nd}$ at random~(see~\cite{B01} for further reference). We will only consider cases where
$nd$ is even, as otherwise there does not exist any $d$-regular graph on $n$ vertices.
In the Configuration Model, the set of vertices in $K_{nd}$ is partitioned into $n$ cells of size $d$ and each cell $W_v$ is
associated to a vertex $v$ of the random regular graph. An edge $e$ of a perfect
matching of $K_{nd}$ induces either a loop in $v$ (if it connects two elements of $W_v$) or an edge
between $v$ and $u$ (if it connects a vertex from $W_v$ to a vertex in $W_u$).

In general, this model may produce graphs with loops and multiple edges.
We will denote by $\mathcal{G}^*(n,d)$ the former probability space and by
$\mathcal{G}(n,d)$ the same probability space conditioned on the event
that $G$ is simple. It is shown in~\cite{MW91} that the following holds:
$$
\Pr\big(G\in\mathcal{G}(n,d)\mid G\in\mathcal{G}^*(n,d) \big)=(1+o(1))e^{\tfrac{1-d^2}{4}}
\qquad\text{if $d=o(\sqrt{n}).$}
$$ Thus, for constant $d$ any property that holds with probability tending to $1$ for
$\mathcal{G}^*(n,d)$ as $n\rightarrow\infty$, will also hold with
probability tending to $1$ for $\mathcal{G}(n,d)$. In this case we
will say that the property holds \emph{with high probability}
(w.h.p.). In fact our bounds include asymptotic terms in $d$, which
means they are meaningful for sufficiently large $d$.
\begin{theorem}
\label{prop:RRG}
	Let $G\in\mathcal{G}(n,d)$ then for any $d\geq 3$,  $\M(G)\leq \tfrac{\log{d}+\log\log
d+O_d(1)}{d}n$ w.h.p..
\end{theorem}
\begin{proof}
First of all we have to show that almost all random regular graphs are twin-free.
	
Observe that the number of perfect matchings of $K_{2m}$ is
$(2m-1)!!=(2m-1)(2m-3)(2m-5)\dots 1$.  Fix a vertex $u$ of $G$ and let
$N(u)=\{ v_1,\dots,v_d \}$. We bound from above the probability that $u$ and
$v_1$ are twins, i.e. $N[u]=N[v_1]$. The number of perfect matchings
of $K_{nd}$ such that in the resulting graph $G$ of
$\mathcal{G}(n,d)$, $v_1$ and $v_2$ are adjacent, is at most
$(d-1)(d-1)(nd-2d-3)!!$. Indeed, there must be an edge between $v_1$
and $v_2$, which gives $(d-1)(d-1)$ possibilities. Since $u$ has $d$
neighbours, the number of possibilities for the remaining graph is the
number of perfect matchings of $K_{nd-2d-2}$.

Analogously the number of perfect matchings with $v_2,v_3\in N(v_1)$
is at most $(d-1)(d-1)(d-2)(d-1)(nd-2d-5)!!$. Thus we have:

\begin{eqnarray*}
\Pr(N[u]=N[v_1]) & \leq &\Pr(N[u]\subseteq N[v_1]) \\
& =&\frac{(d-1)(d-1)(d-2)(d-1)\dots 2(d-1)1(d-1)(nd-4d+1)!!}{(nd-2d-1)!!} \\
& \leq & \frac{d^{d-1}(d-1)!}{(nd-2d-1)\dots (nd-4d+3)} \\
& \leq & \left(\frac{d}{n}\right)^{d-1}\qquad \text{for $n$ large enough.}
\end{eqnarray*}

As we have at most $\tfrac{nd}{2}$ possible pairs of twins (one for
each edge), by the union bound and since $d\geq 3$, for sufficiently large $n$ we obtain:
$$
\Pr(G\mbox{ has twins})\leq \frac{nd}{2}\left(\frac{d}{n}\right)^{d-1}
$$
which tends to $0$ as $n$ tends to infinity.

Therefore, random regular graphs are twin-free w.h.p.	

By~\eqref{eq:expected}, for any $G\in\mathcal{G}(n,d)$,
we have a set $\mathcal{C}$ with
$$
|\mathcal{C}|\leq \frac{\log d+\log\log d+O_d(1)}{d}n
$$ 

that separates any pair of vertices except from the ones where both
vertices belong to a triangle or a 4-cycle. We have to add some vertices
to $\mathcal{C}$ in order to separate the vertices of these small
cycles.
	
Classical results on random regular graphs (independently,~\cite[Corollary 2.19]{B01}
and~\cite{W81}) state that the random variables that count the number
of cycles of length $k$, $X_k$, tend in distribution to independent
Poisson variables with parameter $\lambda_k= \tfrac{1}{2k}(d-1)^k$.

Observe that:
$$
\mathbb{E}(X_3)=\frac{(d-1)^3}{6} \quad \mathbb{E}(X_4)=\frac{(d-1)^4}{8}
$$
i.e. a constant number of triangles and 4-cycles are expected.

Using Markov's inequality we can bound the probability of having too many small cycles:

$$
\Pr(X_3>t)\leq \frac{(d-1)^3}{6t} \qquad \Pr(X_4>t)\leq \frac{(d-1)^4}{8t}
$$ 

Setting $t=\vartheta(n)$, where $\vartheta(n)\to \infty$, the previous probabilities are $o(1)$.
Then w.h.p., we have at most $\vartheta(n)$ cycles
of length~$3$ and $\vartheta(n)$ cycles of length~$4$.

Let $T=\{ u_1,u_2,u_3\}$ be a triangle in $G$. As $d\geq 3$
there exists at least one vertex $v_i$ outside the triangle (moreover, we showed that the graph has
no twins w.h.p.). Since our graph is twin-free, for
each ordered pair $(u_i,u_j)$ there exists some vertex $v_{ij}$, such that $v_{ij}\in N(u_i)\backslash N(u_j)$.
Observe that we can add $v_{12}$, $v_{23}$ and $v_{31}$ to $\mathcal{C}$ and then any pair of vertices from
$T$ will be separated. 

If $T=\{ u_1,u_2,u_3,u_4\}$ induces a $K_4$, each pair of vertices of
$T$ is contained in some triangle and is separated by the last
step. If $T$ induces a 4-cycle, adding $T$ to $\mathcal{C}$ separates
all the elements in $T$. Otherwise, $T$ induces two triangles and
adding $T$ to $\mathcal{C}$ separates the two vertices which have not
been separated in the last step.

After these two steps, we have added at most $7\vartheta(n)$ vertices to
$\mathcal{C}$. Hence, for any $G\in\mathcal{G}(n,d)$ w.h.p. we obtain:

$$
\M(G)\leq \frac{\log{d}+\log\log d+O_d(1)}{d}n + 7\vartheta(n) =\frac{\log{d}+\log\log
d+O_d(1)}{d}n 
$$
Observe that the $\tfrac{O_d(1)}{d}n$ term contains the $7\vartheta(n)$ term.
\end{proof}

Theorem~\ref{prop:RRG}  shows that despite the fact that for any
$d$, we know infinitely many $d$-regular graphs having a very large
identifying code number (e.g. $n-\tfrac{n}{d}$ for the graphs of
Construction~\ref{constr2} of Section~\ref{apdx:constructions}), almost all $d$-regular
graphs have a very small identifying code.

Moreover, $\M(G)$ is concentrated, as the following theorem and its
corollary show. In fact the following result might be already known,
since a similar result is stated for independent dominating sets
in~\cite{HHV09}. However we could not find it in the literature and
decided to give a proof for the sake of completeness.

\begin{theorem}\label{prop:domination}
  Let $G\in\mathcal{G}(n,d)$, then w.h.p. all the dominating sets of
  $G$ have size at least $\tfrac{\log{d}-2\log\log d}{d}n$.
\end{theorem}

\begin{proof}
We will proceed by contradiction.  Given a set of vertices $D$ of
size~$m$, we will compute the probability that $D$ dominates
$Y=V(G)\setminus D$. Recall that $G$ has been obtained from the
configuration model by selecting a random perfect matching of
$K_{nd}$. Let $y\in Y$ fixed, then let $A_y=\{N(D)\cap \{y\} \neq
\emptyset\}$ be the event that $y$ is dominated by $D$. Its
complementary event corresponds to the situation where none of the edges of
the perfect matching of $K_{nd}$ connects the points corresponding to
$y$ to the ones corresponding to any vertex of $D$. 
Define $W_D=\cup_{v\in D} W_v$ as the set of cells corresponding to $D$ in $K_{nd}$. Then for any $v\in
W_D$, the event $B_v$ corresponds to the fact that $v$ is not connected to any point in $W_y$. If
$W_D=\{v_1,\dots,v_{md}\}$,

\begin{eqnarray*}
\Pr(\overline{A_y})&=&\Pr(\cap_{v\in W_D} B_v)\\
&=&\Pr( B_{v_1})\Pr( B_{v_2}\mid B_{v_1})\dots \Pr( B_{v_{md}}\mid \cap_{i=1}^{md-1} B_{v_i})\\
&=&\left(1-\frac{d}{nd-1}\right)\left(1-\frac{d}{nd-3}\right)\dots
\left(1-\frac{d}{nd-(2md-1)}\right)\\ & = & \prod_{i=1}^{md} \left(
1-\frac{d}{nd-(2i-1)}\right)\\ & \geq & \prod_{i=1}^{md} \left(
1-\frac{1}{n-2m}\right)\\
\end{eqnarray*}

Since $1-x=e^{-x+(\log (1-x)+x)}$ (here we take $x=\tfrac{1}{n-2m}$)
and $\log(1-x)+x=O(x^2)$ (by the Taylor expansion of the logarithm in
$x=0$), we obtain:

\begin{eqnarray*}
\Pr(\overline{A_y})& \geq & \exp \left\{- \sum_{i=1}^{md} \frac{1}{n-2m}+O\left(\frac{1}{(n-2m)^2}\right)\right\}\\
& = & \exp \left\{-(1+o(1))\frac{md}{n-2m} \right\}
\end{eqnarray*}

The probability that $D$ is dominating all vertices of $Y= \{y_1,\dots,y_{n-m}\}$ is:
$$
\Pr\left(\cap_{y\in Y} A_y\right)= \Pr\left(A_{y_1}\right)\Pr\left(A_{y_2}\mid A_{y_1}\right)\dots
\Pr\left(A_{y_{n-m}}\mid \cap_{j=1}^{n-m-1} A_{y_j}\right)
$$

We claim that $\Pr\left(A_{y_{i}}\mid \cap_{j=1}^{i-1} A_{y_j}\right)\leq \Pr\left(A_{y_i}\right)$.
Suppose that $y_1,\dots,y_{i-1}$ are dominated. This means that the
corresponding perfect matching of $K_{nd}$ has an edge between one of
the points corresponding to $y_j$ ($1\leq j\leq i-1$) and one of the
points corresponding to the vertices of $D$. The probability that
$y_i$ is not dominated by $D$ is now the probability that none of the
remaining edges of the perfect matching connect any vertex of $D$ with
$y_i$. Hence:

\begin{align*}
\Pr\left(\overline{A_{y_{i}}}\mid \cap_{j=1}^{i-1}
A_{y_j}\right)=&\left(1-\frac{d}{nd-2i+1}\right)\left(1-\frac{d}{nd-2i-1}\right)\dots
\left(1-\frac{d}{nd-2md+1}\right)\\
 \geq & \left(1-\frac{d}{nd-1}\right)\left(1-\frac{d}{nd-3}\right)\dots
\left(1-\frac{d}{nd-2md+1}\right)\\
 = & \Pr(\overline{A_{y_i}})
\end{align*}

By considering the complementary events, $\Pr\left(A_{y_{i}}\mid
\cap_{j=0}^{i-1} A_{y_j}\right)\leq \Pr\left(A_{y_i}\right)$. Hence
these events are negatively correlated, and:

$$
\Pr\left(\cap_{y\in Y} A_y\right)\leq \prod_{i=1}^{n-m} \Pr(A_{y_i}) \leq \left( 1-
e^{-\tfrac{md}{n-2m}}\right)^{n-m}\leq \exp\left\{-(n-m)e^{-\tfrac{md}{n-2m}}\right\} 
$$

For the sake of contradiction, let $m\leq \tfrac{\log{d}-c\log\log{d}}{d}n$ for some $c>2$. Then:

\begin{eqnarray*}
\Pr\left(\cap_{y\in Y} A_y\right)&\leq&\exp\left\{-\left(1-\frac{\log d-c\log\log d}{d}\right)n\exp\left\{-\frac{\log d-c\log\log d}{1-2\tfrac{\log d-c\log\log d}{d}}\right\}\right\}\\
&=&\exp\left\{-\left(1+o_d(1)\right)n\exp\left\{-\frac{\log d-c\log\log d}{1+o_d(1)}\right\}\right\}\\
&=&(1+o_d(1))e^{-\frac{(\log d)^c}{d}n}
\end{eqnarray*}

Note that if no set of size $m$ dominates $Y$, neither will do a smaller one. So we have to look
just at the sets of size $m$.
The number of these sets can be bounded by
\begin{align*}
\binom{n}{m} \leq& \frac{n^{m}}{m!} \leq 
\left(\frac{en}{m}\right)^{m} = \left(
\frac{de}{\log{d}-c\log\log
d}\right)^{\tfrac{\log{d}-c\log\log{d}}{d}n}\\
=&(1+o_d(1))\left(\frac{de}{\log{d}}\right)^{\tfrac{\log{d}-c\log\log{d}}{d}n}
\end{align*}
where we have used $m!\geq \left(\tfrac{m}{e}\right)^m$. 

Let $E_{DS}$ be the event that $G$ has a dominating set of
size $m$. Applying the union bound, we obtain:
\begin{eqnarray*}
\Pr(E_{DS})& \leq & 
(1+o_d(1))\left(\frac{de}{\log{d}}\right)^{\tfrac{\log{d}-c\log\log{d}}{d}n} e^{-\frac{(\log{d})^c}{d}n}  \\
& = & (1+o_d(1))\exp\left\{\frac{\log{d}-c\log\log{d}}{d}(\log{d}+1-\log\log{d})n
-\frac{(\log{d})^c}{d}n \right\}\\
&  = & (1+o_d(1))\exp\left\{ \left(\frac{(\log{d})^2}{d}	
-\frac{(\log{d})^c}{d}+o_d\left(\frac{(\log{d})^2}{d}\right)\right)n \right\} \longrightarrow 0
\end{eqnarray*}
since $c>2$. This shows that w.h.p. no set of size less than $\tfrac{\log{d}-2\log\log{d}}{d}n$ can
dominate the whole graph and completes the proof. 
\end{proof}

Since any identifying code is also a dominating set, we obtain the
following immediate corollary.
\begin{corollary}
  Let $G\in\mathcal{G}(n,d)$, then w.h.p. $\M(G)\geq\tfrac{\log{d}-2\log\log d}{d}n$.
\end{corollary}

Plugging together Theorems~\ref{prop:RRG} and~\ref{prop:domination},
we obtain the following result.

\begin{corollary}
\label{cor:lb}
Let $G\in\mathcal{G}(n,d)$, then w.h.p. 
$$\frac{\log{d}-2\log\log d}{d}n\leq\M(G)\leq\frac{\log{d}+\log\log
  d+O_d(1)}{d}n$$
\end{corollary}

\section{Extremal constructions}~\label{apdx:constructions}

This section gathers some constructions which show the tightness of
some of our upper bounds. Some of these constructions can be found in~\cite{F09}.

\begin{construction}\label{constr1}
Given any $d_H$-regular multigraph $H$ (without loops) on $n_H$ vertices, let $\mathcal{C}_1(H)$ be the graph
on $n=n_H(d_H+1)$ and maximum degree $d=d_H+1$ constructed as follows:
\begin{enumerate}
\item Replace each vertex $v$ of $H$ by a clique $K(v)$ of $d_H+1$ vertices
\item For each vertex $v$ of $H$, let $N(v)=\{v_1,\ldots, v_{d_H}\}$ and
  $K(v)=\{k_0(v),\ldots,k_{d_H}(v)\}$. For each $k_i(v)$ but one ($1\le
  i\le d_H$), connect it with an edge in $\mathcal{C}_1(H)$, to a unique
  vertex of $K(v_i)$, denoted $f\left(k_i(v)\right)$.
\end{enumerate}
\end{construction}

One can see that the graphs $\mathcal{C}_1(H)$ given by
Construction~\ref{constr1} are twin-free. Moreover, for each vertex
$v$ of $H$ and for each $1\le i\le d_H$, note that
$f\left(k_i(v)\right)$ is $k_0(v)k_i(v)$-forced. Therefore
$\mathcal{C}_1(H)$ has $d_Hn_H=n-\tfrac{n}{d}$ forced vertices. In
fact these forced vertices form an identifying code,
therefore $\M(\mathcal{C}_1(H))=n-\tfrac{n}{d}$. An example of this
construction is given in Figure~\ref{fig:constr1}, where $H$ is the
hypercube of dimension~3, $H_3$, and the black vertices are those which
belong to a minimum identifying code of $\mathcal{C}_1(H_3)$.

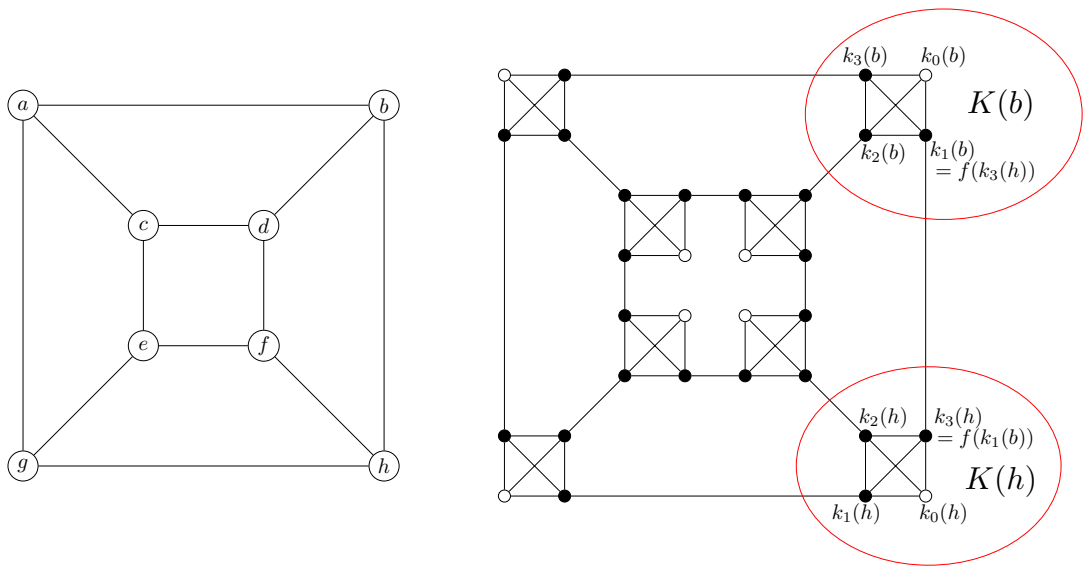
\begin{figure}[ht!]
\centering
\scalebox{0.8}{\begin{tikzpicture}[join=bevel,inner sep=0.5mm,line width=0.4pt, scale=0.5,minimum size=.2cm]
\path (0,1) node[draw,shape=circle,minimum size=.5cm] (v000) {$g$}
      (4,5) node[draw,shape=circle,minimum size=.5cm] (v001) {$e$}
      (12,1) node[draw,shape=circle,minimum size=.5cm] (v100) {$h$}
      (12,13) node[draw,shape=circle,minimum size=.5cm] (v110) {$b$}
      (0,13) node[draw,shape=circle,minimum size=.5cm] (v010) {$a$}
      (8,9) node[draw,shape=circle,minimum size=.5cm] (v111) {$d$}
      (4,9) node[draw,shape=circle,minimum size=.5cm] (v011) {$c$}
      (8,5) node[draw,shape=circle,minimum size=.5cm] (v101) {$f$}
      (16,0) node[draw,shape=circle] (160) {}
      (18,0) node[draw,shape=circle,fill=black] (180) {}
      (16,2) node[draw,shape=circle,fill=black] (162) {}
      (18,2) node[draw,shape=circle,fill=black] (182) {}
      (28,0) node[draw,shape=circle,fill=black] (280) {}
      (28,2) node[draw,shape=circle,fill=black] (282) {}
      (30,2) node[draw,shape=circle,fill=black] (302) {}
      (30,0) node[draw,shape=circle] (300) {}
      (16,12) node[draw,shape=circle,fill=black] (1612) {}
      (18,12) node[draw,shape=circle,fill=black] (1812) {}
      (16,14) node[draw,shape=circle] (1614) {}
      (18,14) node[draw,shape=circle,fill=black] (1814) {}
      (28,12) node[draw,shape=circle,fill=black] (2812) {}
      (28,14) node[draw,shape=circle,fill=black] (2814) {}
      (30,14) node[draw,shape=circle] (3014) {}
      (30,12) node[draw,shape=circle,fill=black] (3012) {}
      (20,4) node[draw,shape=circle,fill=black] (204) {}
      (22,4) node[draw,shape=circle,fill=black] (224) {}
      (20,6) node[draw,shape=circle,fill=black] (206) {}
      (22,6) node[draw,shape=circle] (226) {}
      (24,4) node[draw,shape=circle,fill=black] (244) {}
      (26,4) node[draw,shape=circle,fill=black] (264) {}
      (24,6) node[draw,shape=circle] (246) {}
      (26,6) node[draw,shape=circle,fill=black] (266) {}
      (20,8) node[draw,shape=circle,fill=black] (208) {}
      (22,8) node[draw,shape=circle] (228) {}
      (20,10) node[draw,shape=circle,fill=black] (2010) {}
      (22,10) node[draw,shape=circle,fill=black] (2210) {}
      (24,8) node[draw,shape=circle] (248) {}
      (26,8) node[draw,shape=circle,fill=black] (268) {}
      (24,10) node[draw,shape=circle,fill=black] (2410) {}
      (26,10) node[draw,shape=circle,fill=black] (2610) {}
      (3012)+(2.5,1) node[scale=1.5] {$K(b)$}
      (300)+(2.5,.5) node[scale=1.5] {$K(h)$}
      (3014)+(.6,.6) node {$k_0(b)$}
      (3012)+(.9,-.5) node {$k_1(b)$}
      (3012)+(2.0,-1.2) node {$=f(k_3(h))$}
      (2812)+(.6,-.6) node {$k_2(b)$}
      (2814)+(0,.6) node {$k_3(b)$}
      (300)+(.6,-.6) node {$k_0(h)$}
      (282)+(0.6,.6) node {$k_2(h)$}
      (280)+(-0.3,-.6) node {$k_1(h)$}
      (302)+(1.1,.6) node {$k_3(h)$}
      (302)+(2.0,-.1) node {$=f(k_1(b))$};
\draw (v000) -- (v100) -- (v110) -- (v010) -- (v000) -- (v001) -- (v011) -- (v111) -- (v101) -- (v001)
      (v010) -- (v011)
      (v110) -- (v111)
      (v100) -- (v101)
      (1612) -- (162) -- (160) -- (180) -- (182) -- (162) -- (180) -- (280) -- (300) -- (302) -- (282) -- (280) -- (302)
      (302) -- (3012) -- (3014) -- (2814) -- (2812) -- (3012) -- (2814) -- (1814) -- (1614) -- (1612) -- (1812) -- (1814) -- (1612)
      (208) -- (206) -- (204) -- (224) -- (206) -- (226) -- (224) -- (244) -- (264) -- (266) -- (244) -- (246) -- (266) -- (268)
      (208) -- (228) -- (2210) -- (208) -- (2010) -- (2210) -- (2410) -- (2610) -- (268) -- (248) -- (2410) -- (268)
      (1614) -- (1812) -- (2010) -- (228)
      (160) -- (182) -- (204) -- (226)
      (248) -- (2610) -- (2812) -- (3014)
      (246) -- (264) -- (282) -- (300);
\draw[red] (2812)+(2.6,.7) ellipse (4.6cm and 3.5cm)
           (280)+(2.1,1) ellipse (4.4cm and 3.3cm);
\end{tikzpicture}}
\caption{The graphs $H_3$ and $\mathcal{C}_1(H_3)$}
\label{fig:constr1}
\end{figure}

The following construction is very similar, but yields regular graphs.

\begin{construction}{\cite{F09}}\label{constr2}
Given any $d_H$-regular multigraph $H$ (without loops) on $n_H$ vertices, let
$\mathcal{C}_2(H)$ be the $d$-regular graph on $n=n_Hd_H$ vertices
(where $d=d_H$) constructed as follows:
\begin{enumerate}
\item Replace each vertex $v$ of $H$ by a clique $K(v)$ of $d_H$ vertices.
\item For each vertex $v$ of $H$, let $N(v)=\{v_1,\ldots, v_{d_H}\}$
  and $K(v)=\{k_1(v),\ldots,k_{d_H}(v)\}$. For each $k_i(v)$ ($1\le i\le
  d_H$), connect it with an edge in $\mathcal{C}_2(H)$, to a unique
  vertex of $K(v_i)$, denoted $f\left(k_i(v)\right)$.
\end{enumerate}
\end{construction}

Note that for some vertex $v$ of $H$, in order to separate each pair
of vertices $k_i(v),k_j(v)$ of $K(v)$ in $\mathcal{C}_2(H)$, either
$f\left(k_i(v)\right)$ or $f\left(k_j(v)\right)$ must belong to any
identifying code. Repeating this argument for each pair shows that at
least $d-1$ such vertices are needed in the code. Since for any two
cliques $K(u)$ and $K(v)$, the set of these neighbours are disjoint,
this shows that at least $n_H(d-1)$ vertices are needed in an
identifying code of $\mathcal{C}_2(H)$. In fact it is easy to
construct an identifying code of this size. This shows that despite
the fact that $\mathcal{C}_2(H)$ has no forced vertices,
$\M(\mathcal{C}_2(H))=n-\tfrac{n}{d}$. An example of this
construction is given in Figure~\ref{fig:constr2}, where $H$ is the
complete graph $K_5$, and the black vertices form a minimum
identifying code of $\mathcal{C}_2(K_5)$.

Construction~\ref{constr1} and~\ref{constr2} are close to Sierpi\'nski graphs, which were defined
in~\cite{KM97}. Recently in~\cite{GKMMP12}, it has been shown that Sierpi\'nski graphs are also
extremal with respect to Conjecture~\ref{conj}, i.e. for any Sierpi\'nski graph $G$ on $n$ vertices with maximum degree~$d$, $\M(G)=n-\tfrac{n}{d}$.

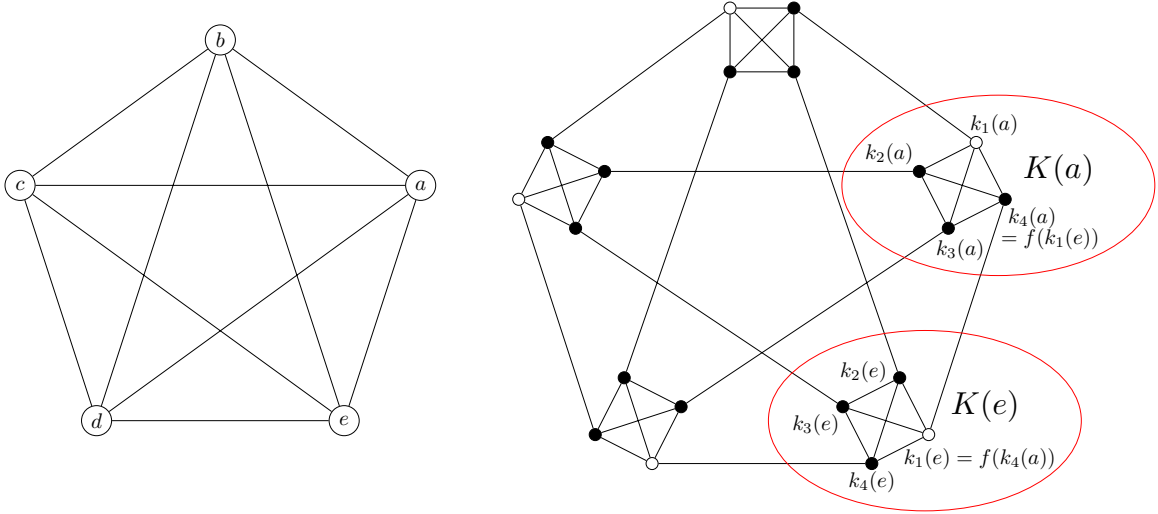
\begin{figure}[ht!]
\centering
\scalebox{0.8}{\begin{tikzpicture}[join=bevel,inner sep=0.5mm,line width=0.4pt, scale=0.5,minimum size=.2cm]
\path (18:7cm) node[draw,shape=circle,minimum size=.5cm] (a) {$a$}
      (18+72:7cm) node[draw,shape=circle,minimum size=.5cm] (b) {$b$}
      (18+2*72:7cm) node[draw,shape=circle,minimum size=.5cm] (c) {$c$}
      (18+3*72:7cm) node[draw,shape=circle,minimum size=.5cm] (d) {$d$}
      (18+4*72:7cm) node[draw,shape=circle,minimum size=.5cm] (e) {$e$}
      (18,0) node (dec) {}
      (dec)+(18:7cm) node (a2) {}
      (dec)+(18+72:7cm) node (b2) {}
      (dec)+(18+2*72:7cm) node (c2) {}
      (dec)+(18+3*72:7cm) node (d2) {}
      (dec)+(18+4*72:7cm) node (e2) {}

      (a2)+(72:1.5cm) node[draw,shape=circle] (a11) {}
      (a2)+(72+90:1.5cm) node[draw,shape=circle,fill=black] (a12) {}
      (a2)+(72+2*90:1.5cm) node[draw,shape=circle,fill=black] (a21) {}
      (a2)+(72+3*90:1.5cm) node[draw,shape=circle,fill=black] (a22) {}
      (b2)+(45:1.5cm) node[draw,shape=circle,fill=black] (b11) {}
      (b2)+(45+90:1.5cm) node[draw,shape=circle] (b12) {}
      (b2)+(45+2*90:1.5cm) node[draw,shape=circle,fill=black] (b21) {}
      (b2)+(45+3*90:1.5cm) node[draw,shape=circle,fill=black] (b22) {}
      (c2)+(18:1.5cm) node[draw,shape=circle,fill=black] (c11) {}
      (c2)+(18+90:1.5cm) node[draw,shape=circle,fill=black] (c12) {}
      (c2)+(18+2*90:1.5cm) node[draw,shape=circle] (c21) {}
      (c2)+(18+3*90:1.5cm) node[draw,shape=circle,fill=black] (c22) {}
      (d2)+(18:1.5cm) node[draw,shape=circle,fill=black] (d11) {}
      (d2)+(18+90:1.5cm) node[draw,shape=circle,fill=black] (d12) {}
      (d2)+(18+2*90:1.5cm) node[draw,shape=circle,fill=black] (d21) {}
      (d2)+(18+3*90:1.5cm) node[draw,shape=circle] (d22) {}
      (e2)+(72:1.5cm) node[draw,shape=circle,fill=black] (e11) {}
      (e2)+(72+90:1.5cm) node[draw,shape=circle,fill=black] (e12) {}
      (e2)+(72+2*90:1.5cm) node[draw,shape=circle,fill=black] (e21) {}
      (e2)+(72+3*90:1.5cm) node[draw,shape=circle] (e22) {}

      (a2)+(3.2,0.5) node[scale=1.5] {$K(a)$}
      (e2)+(3.3,0.5) node[scale=1.5] {$K(e)$}
      (a11)+(.6,.6) node {$k_1(a)$}
      (a12)+(-1.0,.6) node {$k_2(a)$}
      (a21)+(.4,-.7) node {$k_3(a)$}
      (a22)+(0.9,-.6) node {$k_4(a)$}
      (a22)+(1.5,-1.3) node {$=f(k_1(e))$}
      (e11)+(-1.2,0.2) node {$k_2(e)$}
      (e12)+(-0.9,-0.6) node {$k_3(e)$}
      (e21)+(0,-.6) node {$k_4(e)$}
      (e22)+(1.7,-.8) node {$k_1(e)=f(k_4(a))$}
;
 
\draw (a) -- (b) -- (c) -- (d) -- (e) -- (a) -- (c) -- (e)
      (a) -- (d) -- (b) -- (e)
      (b12) -- (b21) -- (b11) -- (b22) -- (b12) -- (b11) -- (a11) -- (a12) -- (a22) -- (a21) -- (a11) -- (a22)
      (a22) -- (e22) -- (e11) -- (e21) -- (e22) -- (e12) -- (e21) -- (d22) -- (d11) -- (d21) -- (d22) -- (d12) -- (d21)
      (d21) -- (c21) -- (c22) -- (c12) -- (c21) -- (c11) -- (c12) -- (b12)
      (b21) -- (b22) -- (e11) -- (e12) -- (c22) -- (c11) -- (a12) -- (a21) -- (d11) -- (d12) -- (b21);

\draw[red] (a2)+(1.2,0) ellipse (5.2cm and 3cm)
           (e2)+(1.3,0) ellipse (5.2cm and 3cm);
\end{tikzpicture}}
\caption{The graphs $K_5$ and $\mathcal{C}_2(K_5)$}
\label{fig:constr2}
\end{figure}

\begin{construction}{\cite{F09}}\label{constr3}
  Given an even number $2k$ and an integer $d\ge 3$, we construct a
  twin-free $d$-regular triangle-free graph $\mathcal{C}_3(2k,d)$ on
  $n=2kd$ vertices as follows.
\begin{enumerate}
\item Let $\{c_0,\ldots,c_{2k-1}\}$ be a set of $2k$ vertices and add the
  edges of the perfect matching $\{c_ic_{i+1\bmod 2k} \mid i \mbox{ is odd}\}$.
\item For each even $i$ ($0\leq i\leq 2k-2$), build a copy $K(i)$ of the
  complete bipartite graph $K_{d-1,d-1}$. Join vertex $c_i$ to all
  vertices of one part of the bipartition of $K(i)$, and join vertex
  $c_{i+1}$ to all other vertices of $K(i)$.
\end{enumerate}
\end{construction}

Consider an identifying code of $\mathcal{C}_3(2k,d)$. Note that in
each copy $K(i)$ of $K_{d-1,d-1}$, at least $2d-4$ vertices belong to
the code in order to separate the vertices being in the same part of
the bipartition of $K(i)$. Now if exactly $2d-4$ vertices of $K(i)$
belong to the code, in order to separate the two remaining vertices,
either $c_i$ or $c_{i+1}$ belongs to the code. Hence for each odd $i$,
at most three vertices from $\{c_i,c_{i+1}\}\cup V(K(i))$ do not belong
to a code of $\mathcal{C}_3(2k,d)$. On the other hand, taking all
vertices $c_i$ such that $i$ is even together with $d-2$ vertices of
each part of the bipartition of each copy of $K_{d-1,d-1}$ yields an
identifying code of this size. Hence
$\M(\mathcal{C}_3(2k,d))=k+2k(d-2)=n-\tfrac{n}{2d/3}$. An example
of this construction is given in Figure~\ref{fig:constr3}, where
$2k=8$, $d=3$, and the black vertices form a minimum identifying code
of $\mathcal{C}_3(8,3)$.

\begin{figure}[ht!]
\centering
\scalebox{0.7}{\begin{tikzpicture}[join=bevel,inner sep=0.5mm,line width=0.6pt, scale=0.5,minimum size=.2cm]
\path (0,0) node[draw,shape=circle] (c1) {}
      (4,0) node[draw,shape=circle,fill=black] (c2) {}
      (8,0) node[draw,shape=circle] (c3) {}
      (12,0) node[draw,shape=circle,fill=black] (c4) {}
      (16,0) node[draw,shape=circle] (c5) {}
      (20,0) node[draw,shape=circle,fill=black] (c6) {}
      (24,0) node[draw,shape=circle] (c7) {}
      (28,0) node[draw,shape=circle,fill=black] (c8) {}
      (c1)+(0,1) node[scale=1.2] {$c_0$}
      (c2)+(0,1) node[scale=1.2] {$c_1$}
      (c3)+(0,1) node[scale=1.2] {$c_2$}
      (c4)+(0,1) node[scale=1.2] {$c_3$}
      (c5)+(0,1) node[scale=1.2] {$c_4$}
      (c6)+(0,1) node[scale=1.2] {$c_5$}
      (c7)+(0,1) node[scale=1.2] {$c_6$}
      (c8)+(0,1) node[scale=1.2] {$c_7$}
      (c1)+(0,-4) node[draw,shape=circle,fill=black] (c11) {}
      (c1)+(0,-8) node[draw,shape=circle] (c12) {}
      (c2)+(0,-4) node[draw,shape=circle,fill=black] (c21) {}
      (c2)+(0,-8) node[draw,shape=circle] (c22) {}
      (c3)+(0,-4) node[draw,shape=circle,fill=black] (c31) {}
      (c3)+(0,-8) node[draw,shape=circle] (c32) {}
      (c4)+(0,-4) node[draw,shape=circle,fill=black] (c41) {}
      (c4)+(0,-8) node[draw,shape=circle] (c42) {}
      (c5)+(0,-4) node[draw,shape=circle,fill=black] (c51) {}
      (c5)+(0,-8) node[draw,shape=circle] (c52) {}
      (c6)+(0,-4) node[draw,shape=circle,fill=black] (c61) {}
      (c6)+(0,-8) node[draw,shape=circle] (c62) {}
      (c7)+(0,-4) node[draw,shape=circle,fill=black] (c71) {}
      (c7)+(0,-8) node[draw,shape=circle] (c72) {}
      (c8)+(0,-4) node[draw,shape=circle,fill=black] (c81) {}
      (c8)+(0,-8) node[draw,shape=circle] (c82) {}
      (c11)+(-2,-2) node[scale=1.5] {$K(0)$};
\draw (c1) -- (c21) -- (c12) -- (c22) -- (c11) -- (c21)
      (c4) -- (c31) -- (c41) -- (c32) -- (c42) -- (c31)
      (c5) -- (c61) -- (c52) -- (c62) -- (c51) -- (c61)
      (c8) -- (c71) -- (c81) -- (c72) -- (c82) -- (c71)
      (c1) -- (c22)
      (c12) -- (c2)
      (c11) -- (c2) -- (c3) -- (c41)
      (c3) -- (c42)
      (c32) -- (c4) -- (c5)
      (c5) -- (c62)
      (c52) -- (c6)
      (c51) -- (c6) -- (c7) -- (c81)
      (c7) -- (c82)
      (c72) -- (c8)
      (c1) .. controls +(-6,4) and +(6,4) .. (c8);
\draw[red] (c11)+(1.5,-2) ellipse (5cm and 3cm);
\end{tikzpicture}}
\caption{The graph $\mathcal{C}_3(8,3)$}
\label{fig:constr3}
\end{figure}
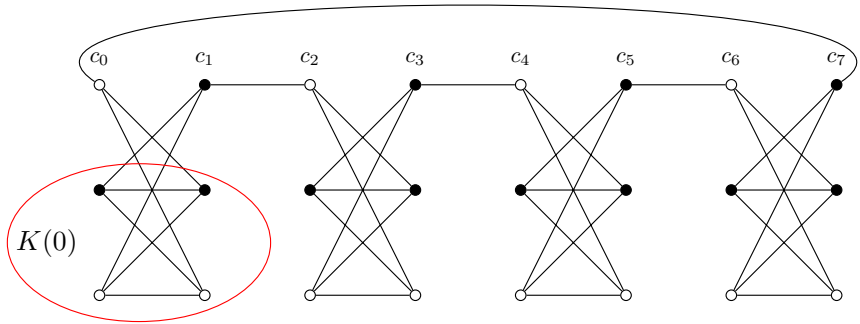

~\\~\\
\vspace{0.3cm} \noindent\textbf{Acknowledgements}\\ The authors are
thankful to the referee for his careful reading and very detailed
comments, which have helped them improving the quality of this paper.

\end{document}